\documentclass[a4paper]{scrartcl} 
\usepackage[english]{babel}
\usepackage[intlimits]{amsmath}
\usepackage[pdftex,urlcolor=blue]{hyperref}
\usepackage{amssymb}
\usepackage{amsthm} 
\usepackage{graphicx}

\graphicspath{{bock_et_al---2009---figures/}}

\newtheorem{proposition}{Proposition} 
\newtheorem{lemma}{Lemma} 
\newtheorem{definition}{Definition} 

\newcommand{\eg}{\mbox{e.g.} }
\newcommand{\ie}{\mbox{i.e.} }
\newcommand{\cf}{\mbox{cf.} }

\newcommand{\vx}{\ensuremath{\mathbf{x}}}
\newcommand{\vxi}{\ensuremath{\mathbf{x}_i}}
\newcommand{\vxj}{\ensuremath{\mathbf{x}_j}}

\newcommand{\vxij}{\ensuremath{\mathbf{x}_{ij}}}
\newcommand{\udij}{\ensuremath{\hat{\mathbf{d}}_{ij}}}

\newcommand{\vMij}{\ensuremath{\mathbf{M}_{ij}}}
\newcommand{\vSij}{\ensuremath{\mathbf{S}_{ij}}}
\newcommand{\vA}{\ensuremath{\mathbf{A}}}
\newcommand{\vT}{\ensuremath{\mathbf{T}}}
\newcommand{\vCT}{\ensuremath{\mathbf{C}_T}}

\newcommand{\uRif}{\ensuremath{\hat{\mathbf{R}}_{i0}}}
\newcommand{\uRi}{\ensuremath{\hat{\mathbf{R}}_{i}}}
\newcommand{\uRj}{\ensuremath{\hat{\mathbf{R}}_{j}}}

\newcommand{\vvi}{\ensuremath{\mathbf{v}_i}}

\newcommand{\vFi}{\ensuremath{\mathbf{F}_i}}
\newcommand{\vFj}{\ensuremath{\mathbf{F}_j}}
\newcommand{\vFiji}{\ensuremath{\mathbf{F}_{ij}^{\mathrm{(int)}}}}
\newcommand{\vFjii}{\ensuremath{\mathbf{F}_{ji}^{\mathrm{(int)}}}}
\newcommand{\Fijhor}{\ensuremath{F_{ij}^{\mathrm{(hor)}}}}
\newcommand{\Fjihor}{\ensuremath{F_{ji}^{\mathrm{(hor)}}}}
\newcommand{\Fijver}{\ensuremath{F_{ij}^{\mathrm{(ver)}}}}
\newcommand{\Fjiver}{\ensuremath{F_{ji}^{\mathrm{(ver)}}}}
\newcommand{\vFifr}{\ensuremath{\mathbf{F}_{i}^{\mathrm{(drag)}}}}
\newcommand{\vFifb}{\ensuremath{\mathbf{F}_{i}^{\mathrm{(loc)}}}}
\newcommand{\vFjfb}{\ensuremath{\mathbf{F}_{j}^{\mathrm{(loc)}}}}
\newcommand{\vFist}{\ensuremath{\mathbf{F}_{i}^{\mathrm{(st)}}}}

\newcommand{\finth}{\ensuremath{f_\mathrm{int}}}

\newcommand{\floc}{\ensuremath{f_\mathrm{loc}}}

\newcommand{\vffi}{\ensuremath{\mathbf{f}_i}}
\newcommand{\vffj}{\ensuremath{\mathbf{f}_j}}
\newcommand{\vffij}{\ensuremath{\mathbf{f}_{ij}}}

\newcommand{\Rij}{\ensuremath{R_{ij}}}
\newcommand{\dij}{\ensuremath{d_{ij}}}
\newcommand{\Rif}{\ensuremath{ R_{i0} }} 
\newcommand{\Rjf}{\ensuremath{ R_{j0} }}
\newcommand{\Rkf}{\ensuremath{ R_{k0} }}
\newcommand{\RiT}{\ensuremath{ R_{iT} }}
\newcommand{\RjT}{\ensuremath{ R_{jT} }}
\newcommand{\deltij}{\ensuremath{\delta_{ij}}}
\newcommand{\deltif}{\ensuremath{\delta_{i0}}}
\newcommand{\deltjf}{\ensuremath{\delta_{j0}}}
\newcommand{\deltmi}{\ensuremath{\delta_{ij}^{\mathrm{(min)}}}}
\newcommand{\deltcr}{\ensuremath{\delta_{ij}^{\mathrm{(crit)}}}}
\newcommand{\deltru}{\ensuremath{\delta_{ij}^{\mathrm{(rup)}}}}
\newcommand{\zzr}{\ensuremath{z_\mathrm{max}}}

\newcommand{\mcVi}{\ensuremath{\mathcal{V}_i}}
\newcommand{\mcVj}{\ensuremath{\mathcal{V}_j}}
\newcommand{\mcVic}{\ensuremath{\overline{\mathcal{V}}_i}}
\newcommand{\mcVjc}{\ensuremath{\overline{\mathcal{V}}_j}}

\newcommand{\mcP}{\ensuremath{\mathcal{P}}}
\newcommand{\mcPm}{\ensuremath{\mathcal{P}_{\mathrm{max}}}}
\newcommand{\mcPi}{\ensuremath{\mathcal{P}_i}}
\newcommand{\mcPj}{\ensuremath{\mathcal{P}_j}}
\newcommand{\mcPk}{\ensuremath{\mathcal{P}_k}}

\newcommand{\smcPm}{\ensuremath{\sqrt{\mathcal{P}_{\mathrm{max}}}}}

\newcommand{\mcBri}{\ensuremath{ \mathcal{B}_{r_i} (\vxi) }}
\newcommand{\mcBrj}{\ensuremath{ \mathcal{B}_{r_j} (\vxj) }}
\newcommand{\mcBRif}{\ensuremath{ \mathcal{B}_{ R_{i0} } (\vxi) }}
\newcommand{\mcBRjf}{\ensuremath{ \mathcal{B}_{ R_{j0} } (\vxj) }}

\DeclareMathOperator{\Var}{Var}

\newcommand{\rhon}{\ensuremath{ \tilde{\rho} }}

\newcommand{\rhofi}{\ensuremath{ \rho_{i0} }}

\newcommand{\thm}{\ensuremath{\theta_T}}

\newcommand{\phif}{\ensuremath{\phi_{i0}}}
\newcommand{\phjf}{\ensuremath{\phi_{j0}}}

\newcommand{\ddd}{\ensuremath{\mathrm{d}}}
\newcommand{\Delmi}{\ensuremath{\Delta_\mathrm{min}}}
\newcommand{\Delcr}{\ensuremath{\Delta_\mathrm{crit}}}
\newcommand{\Delij}{\ensuremath{\Delta_{ij}}}

\newcommand{\qqq}{\ensuremath{Q}}
\newcommand{\qnb}{\ensuremath{Q_\mathrm{nb}}}

\newcommand{\lamwid}{\ensuremath{\mathcal W}}
\newcommand{\covl}{\ensuremath{\mathcal Z} }
\newcommand{\sprod}{\ensuremath{\mathcal T_0}}

\begin{document}

\title{\centering
  Generalized Voronoi Tessellation as a Model of Two-dimensional 
  Cell Tissue Dynamics
}
%
%
\makeatletter
\renewcommand{\@fnsymbol}[1]{\@arabic{#1}}
\makeatother
\author{
  Martin Bock\thanks{ corresponding author, email \texttt{mab@uni-bonn.de}
      } $^,$\thanks{ Universit\"at Bonn, Theoretische Biologie,
                     Kirschallee 1-3, 53115 Bonn, Germany } \and
  Amit Kumar Tyagi\footnotemark[2] \and
  Jan-Ulrich Kreft\thanks{ Centre for Systems Biology,
      School of Biosciences, University of Birmingham, United Kingdom } \and
  Wolfgang Alt\footnotemark[2]
}
\date{\today}
\maketitle

\begin{abstract}
  Voronoi tessellations have been used to model the geometric 
  arrangement of cells in morphogenetic or cancerous tissues, 
  however so far only with flat hypersurfaces as cell-cell contact borders.
  In order to reproduce the experimentally observed 
  piecewise spherical boundary shapes, we develop a consistent 
  theoretical framework of multiplicatively weighted 
  distance functions, defining generalized finite Voronoi 
  neighborhoods around cell bodies of varying 
  radius, which serve as heterogeneous generators of the 
  resulting model tissue.
  The interactions between cells are represented by adhesive 
  and repelling force densities on the cell contact borders. 
  In addition, protrusive locomotion forces are implemented along 
  the cell boundaries at the tissue margin, 
  and stochastic perturbations allow for non-deterministic motility effects. 
  Simulations of the emerging system of stochastic differential 
  equations for position and velocity of cell centers show 
  the feasibility of this Voronoi method generating realistic cell shapes. 
  In the limiting case of a single cell pair in brief contact, 
  the dynamical nonlinear Ornstein-Uhlenbeck 
  process is analytically investigated. 
  In general, topologically distinct tissue conformations are observed,
  exhibiting stability on different time scales, and 
  tissue coherence is quantified by suitable characteristics. 
  Finally, an argument is derived 
  pointing to a tradeoff in natural tissues 
  between cell size heterogeneity and 
  the extension of cellular lamellae.
\end{abstract}

\tableofcontents

\section{Introduction}

A Voronoi tessellation is a partition of space according to certain 
neighborhood relations of a given set of generators (points) in 
this space. Initially proposed by Dirichlet 
for special cases \cite{Dirichlet1850}, the method was established by 
Voronoi more than 100 years ago \cite{Voronoi1908}. 
The geometric dual of the Voronoi tessellation
was proposed by Delaunay in 1934 --- and therefore is called 
Delaunay triangulation. It connects those points of 
a Voronoi tessellation that share a common 
border.
Since the latter can be directly constructed out of the former, 
both terms are sometimes used equivalently. 
In the following years, the method was rediscovered 
throughout other fields, which accounts for many other names 
designating the very concept, such as 
Thiessen polygons \cite{Thiessen1911} in meteorology or 
Wigner-Seitz cells \cite{WignerSeitz1933} 
and Brillouin zones \cite{Brillouin1930} in solid state physics.
With the technological and scientific advance, 
the method became feasible in 
computational geometry \cite{Shamos1975}, and since then has 
widely evolved, \cf \cite{Bernal1992,AurenhammerKlein1996}, 
making it appealing for biological applications.

In particular, Voronoi tessellations have been applied to 
represent various aggregates of cells and swarming animals. 
Initially, Honda proposed the method 
in two spatial dimensions \cite{Honda1978}. 
The first applications to biological tissue were 
cell sorting simulations, however 
starting from artificially shaped quadratic cells \cite{Sulsky1984}. 
Then, morphogenesis and its 
underlying intercellular mechanisms were studied starting from 
a pure Delaunay mesh and simulating vertex dynamics 
\cite{WelikyOster1990,Weliky1991}, yet without 
using Voronoi tessellations explicitly.
In contrast, by applying transformation rules like mitosis 
combined with Monte-Carlo dynamics, evolving multicellular tissue 
was represented by Voronoi tessellations \cite{Drasdo1995}. 
In particular, growth instabilities, blastula formation 
and gastrulation could be conceived within 
this framework \cite{DrasdoForgacs2000}. 
Similar effects were reproduced by using vertex dynamics \cite{Brodland2002}.
Moreover, cell organization 
in the intestinal crypt was modelled using spring forces 
and restricting the motion to a cylindrical surface \cite{Meineke2001}. 
An application to bird swarming together with the proposal of 
a continuum formulation was given in \cite{Alt2003}. 
Finally, the influence of shear stress on the evolution of 
two-dimensional tissues was studied \cite{Dieterich2004}.

Only quite recently Voronoi tessellations
have been extended to be used as a model for three dimensional 
tissue, again using vertex dynamics \cite{Honda2004}.
Other authors use optimized kinetic algorithms \cite{Schaller2004,Beyer2005} 
to employ generalized Voronoi tessellations (discussed as 
difference method in this article), with cell-cell 
and cell-matrix adhesion \cite{Schaller2005}.
Marginal cells have been closed by 
prescribing a maximal cell radius, enabling the 
study of the growth dynamics of epithelial 
cell populations \cite{Galle2005}. 
So far, however, the cell-cell boundaries 
were exclusively represented by flat hypersurfaces. 

In contrast, when observing two-dimensional monolayers 
of keratinocytes, for example \cf \cite{Young2000,Marie2003,Tinkle2008}, 
the cell-cell contact borders visible from staining 
cadherin-complexes frequently appear as circular arcs, whose 
shape and length seems to be determined by the constellation 
and size distribution of neighboring cells.
Moreover, the forces between such cells are 
influenced by filament networks or bundles 
meeting at the cell-cell junctions and eventually 
balanced by elastic counterforces 
\cite{Alberts,Ananthakrishnan2007,Koestler2008,Sivaramakrishnan2008}.
Therefore, a geometrical and dynamical modeling framework 
is required that reproduces the observed cell shapes and 
simultaneously allows for quantifying the cell-cell 
interaction forces as well as the active locomotion 
forces appearing at the free cell boundaries. 
Here we present a simple and effective solution of this 
task by using a suitably weighted Voronoi tessellation.

This article is organized as follows:
In section \ref{sec:voronoi_tesselation} Voronoi tessellations
are introduced in a general manner. 
Next, two types of weighted square distance functions are introduced, 
using the method of difference and quotient, respectively, and 
their particular consequences for cell tissue modeling are investigated. 
Inspired by the intricate interplay between cytoskeletal filament bundles and 
cadherin-catenin cohesion or integrin adhesion sites, 
the forces on the intercellular and exterior cell borders  
are proposed in section \ref{sec:forces} 
after discussing the emergence of cell shape within our model. 
Then the dynamics 
of a whole cell aggregate is defined, directly leading to 
analytical results on cell pair contacts in section \ref{sec:pair_contacts}.
After simulation studies of meta-stable states during tissue 
equilibration and robustness of tissue formation under 
the influence of various model parameters in section \ref{sec:sim} 
we conclude with a discussion of our results in section \ref{sec:results}.

\section{Generalized Voronoi tessellations}\label{sec:voronoi_tesselation}

Let $\{ g_i : i = 1 \dots N \}$ denote a finite 
set of $N$ generators or 
points $\vxi$ in Euclidean, 
$n$-dimensional space $\mathbb{R}^n$. 
\begin{definition}\label{def:Vi}
  The Voronoi cell of a generator $g_i = \vxi$ is defined as
  \begin{equation}\label{eq:VNG}
  \mcVi = \left\{
      \vx \in \mathbb{R}^n :
        \mcPi (\vx) < \mcPj (\vx) 
      \quad \forall j \neq i
    \right\},
  \end{equation}
  where $\mcPi$ $(i = 1 \dots N)$ is a given set of continuous, generalized
  square distance functions on $\mathbb{R}^n$ 
  with the property that $\forall i: \; \vxi \in \mcVi$.
\end{definition}\label{defi:gen_voronoi}
Thus, $\mcVi$ represents an open neighborhood of $\vxi$, 
containing all points $\vx$ that are $\mcP$-closer 
to $\vxi$ than to any other $\vxj$.
\begin{definition}\label{defi:gen_vor_boundary}
  The contact border between two points $\vxi$ and 
  $\vxj$ is defined as the intersection of the closures of 
  $\mcVi, \mcVj$:
  \begin{equation}\label{BG}
    \Gamma_{ij} = \mcVic \cap \mcVjc \quad
        \text{with} \; \; i \neq j.
  \end{equation}
  The total boundary of the Voronoi neighborhood around $\vxi$ then is
  \begin{equation*}
    \partial \mcVi = \bigcup_{j \neq i} \Gamma_{ij} \; .
  \end{equation*}
\end{definition}
The contact border 
$\Gamma_{ij}$ therefore is the set of all points $\mcP$-equidistant
from $\vxi$ and $\vxj$, namely
\begin{equation}\label{eq:pi_eq_pj}
  \Gamma_{ij} = \left\{ \vx \in \mathbb{R}^n : 
                        \mcPi (\vx) = \mcPj (\vx) \leq \mcPk (\vx) \;
                        \forall k \neq i,j
                \right\}.
\end{equation}
The Voronoi tessellation in general form is then given by 
$\{ \mcVi, \Gamma_{ij}; \; i,j = 1, \dots, N \}$ and covers 
the whole space, where so-called marginal 
neighborhoods $\mathcal{V}_i$ extend to infinity.
Depending on the particular choice of the generalized square 
distance functions, 
$\Gamma_{ij}$ can take various shapes. In the standard Euclidean case, 
$\mathcal{P}_i (\mathbf{x}) = | \mathbf{x} - \mathbf{x}_i |^2$, 
the contact border $\Gamma_{ij}$ is the perpendicular 
bisector of the line segment from $\vxi$ to $\vxj$, an $(n-1)$-plane. 
Then $\mathcal{V}_i$ is bounded by a convex, not necessarily finite 
polytope and is called the (classical) Voronoi neighborhood \cite{Voronoi1908} 
or Dirichlet domain \cite{Dirichlet1850}.

The modeling aim here is to represent biological eukaryotic cells 
in connected $3$-dimensional tissues or confluent $2$-dimensional 
cell monolayers as Voronoi neighborhoods $\mcVi$, as in the 
$2$-dimensional pictures in figure \ref{fig:stratum_spinosum}. 
\begin{figure}
  \centering
  \includegraphics[width=1.0\textwidth]{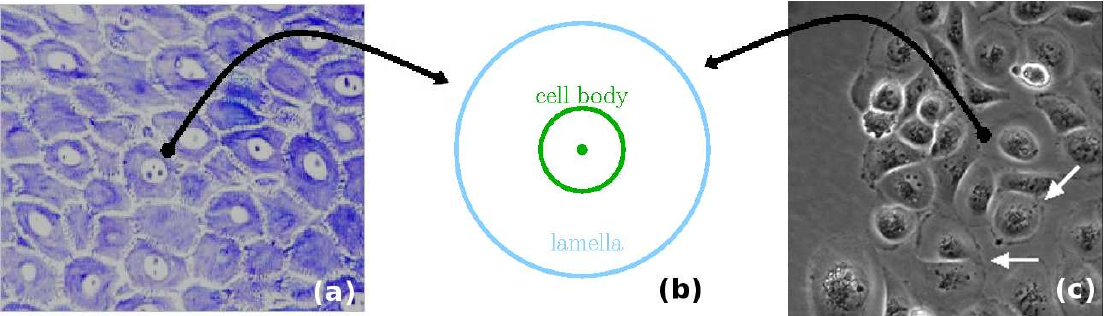}
  \caption{(a) Micrograph of human
           keratinocytes in a section of stratum spinosum \emph{in vivo} 
           (reproduced from \cite{Young2000}), 
           and (c) phase contrast microscopic 
           photograph of human keratinocytes in-vitro 
           (by courtesy of Institute of Cell Biology, Bonn University). 
           Image (c) is extracted from the supplementary movie 
           \href{http://www.theobio.uni-bonn.de/people/mab/sup01/index.html}{\texttt{mov0.avi}}
           of the dynamics at the edge of an almost confluent monolayer 
           in a wound scratch assay.
           White arrows in (c) indicate a round 
           cell (upper) and a cell pair competing for its influence 
           region (lower).
           According to these observations,  
           tissue cells in a two-dimensional geometry are modelled as 
           Voronoi neighborhoods $\mcVi$ containing the 
           ball shaped cell bodies $\mcBri$ surrounded by a so-called 
           lamella (b). 
          }\label{fig:stratum_spinosum}
\end{figure}
In a minimal approach we define the points $\vxi$ as
centers of the visible, mostly ball-shaped cell bodies, containing the 
cell nuclei plus other cell organelles such as mitochondria 
or the Golgi apparatus. 
By attributing a finite radius $r_i > 0$ to each $\vxi$, 
the Voronoi concept is extended to generators $g_i = \mcBri$ 
of positive finite volume, being  
a suitable representation of \emph{cell bodies}. 
Since these are rather solid in comparison to the rest of the cell,
it is assumed that the $\mcBri$ do not overlap. Then 
$\mcVi \setminus \mcBri$ represents the \emph{protoplasmic region} of the 
cell $i$, which for $n=2$ appears as a flat \emph{lamella} in light microscopy.
Clearly, the natural condition $\overline{\mcBri} \subset \mcVi$ requires that 
\begin{equation}\label{eq:nuclei_non_overlap}
  \forall \vx \in \overline{\mcBri} \quad
  \forall j \neq i: \quad
  \mcPj (\vx) > \mcPi (\vx).
\end{equation}
It shall be seen later, that this condition is fulfilled for the 
chosen generalized square distance functions.

Furthermore, certain weights $w_i \in \mathbb{R}^+$, 
on which the square distance function $\mcPi$ may depend, 
are assigned to each cell $i$.
These weights $w_i = w (r_i)$ are assumed to be 
strictly monotonically increasing 
functions of the cell body radius $r_i$, 
with the intended effect that 
stronger weights induce larger cell sizes, 
by shifting the Voronoi contact borders outwards. 
Importantly, different choices of how $\mcPi$ depends on $w_i$ 
could lead to different cell shapes.
Out of many possible generalizations of Voronoi tessellations
(for a review see \cite{AurenhammerKlein1996}), 
we only discuss two straight-forward ways here, which are determined by 
the set of all cell center positions, body radii and weights 
$\lbrace g_i = (\mcBri, w_i) \rbrace$.

\subsection{Difference method} 

The partition of space into cells is obtained by a Voronoi tessellation 
using the Euclidean square distance function with subtracted weights
\begin{equation}\label{eq:diffrule}
  \mcPi (\vx) = ( \vx - \vxi )^2 - w_i^2,
\end{equation}
which has previously been used in \cite{Honda1978,Honda2004} 
without and in \cite{Schaller2005,Beyer2007} with weights $w_i$. 

For the following we denote 
$\vxij = ( \vxi + \vxj )/2$ the cell pair mid-point, 
$\dij = | \vxi - \vxj |$ the Euclidean cell center distance, 
and $\udij = ( \vxi - \vxj ) / \dij$ the unit vector 
of the oriented axis connecting the two cell centers. 
Thus, from equation (\ref{eq:pi_eq_pj}) the condition for a point
$\vx$ to be located on the contact border $\Gamma_{ij}$ 
reads as 
\begin{equation}\label{eq:diffrule_boundary}
  \bigl( \vx - \vxij \bigr) \cdot \udij = - \frac{ w_i^2 - w_j^2}{ 2 \dij }, 
\end{equation}
being equivalent to a linear hyper-plane equation.
As for the classical Voronoi partition, 
the contact $(n-1)$-plane between two
neighboring cells $i,j$ is perpendicular to the
vector connecting the cell centers. However, now the position 
of the contact border plane along the connecting vector, 
and thereby the sizes of the Voronoi cells, 
depend on the weights.
\begin{figure}[htb]
  \centering
  \includegraphics[width=0.6\textwidth]{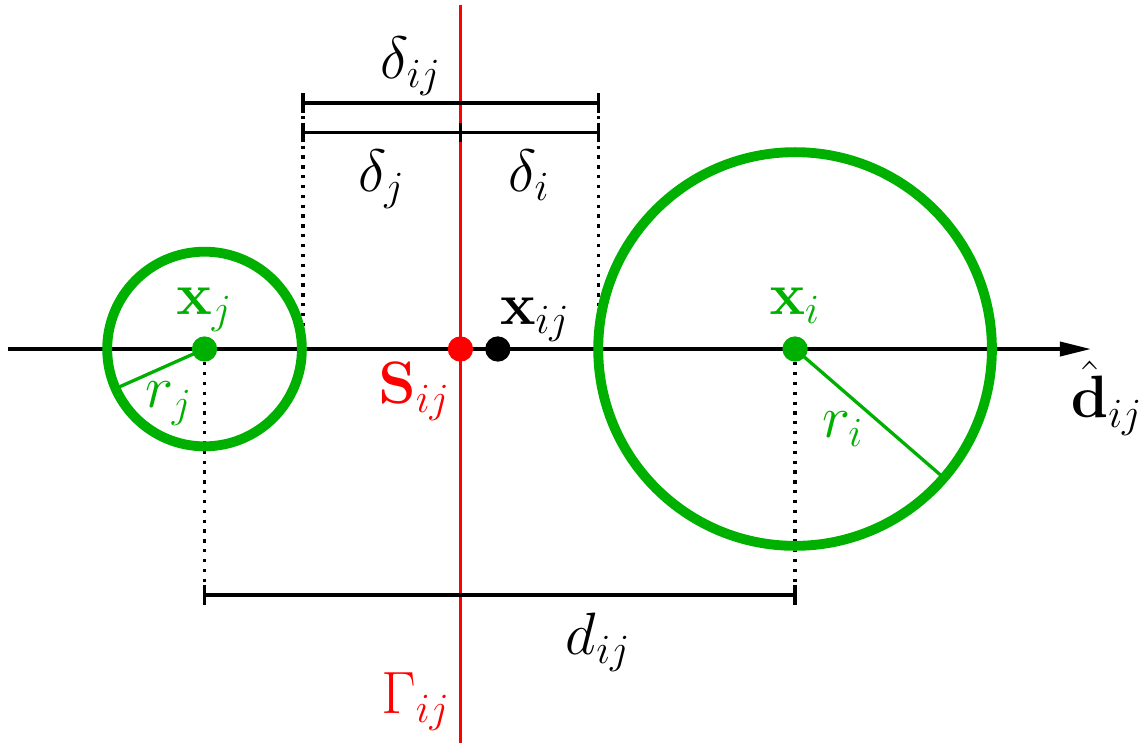}
    \caption{Geometry of a cell pair and its border using the
             difference method. There is a distinct direction $\udij$ 
             given by the axis line connecting $\vxj$ and $\vxi$. 
             The circular cell bodies around $\vxi, \vxj$ 
             and their radii $r_i, r_j$ are indicated 
             by green lines, the contact hyperplane $\Gamma_{ij}$ 
             by a red one.
             Other quantities are explained in the text. 
            }\label{fig:diffrule}
\end{figure}
In figure \ref{fig:diffrule} the geometry of a cell pair 
with its separating border is illustrated.  

In order for a distance function to yield a partition 
describing an aggregate of biological cells in living tissue, 
the border $\Gamma_{ij}$ has to be located 
between the surfaces of the non-overlapping cell bodies.
With the corresponding distances as denoted in figure \ref{fig:diffrule}, 
this is equivalent to
\begin{equation}\label{eq:drCond0}
  \delta_i > 0 \quad \wedge \quad \delta_j > 0.
\end{equation}
These constraints have
consequences for possible choices of the weights:
\begin{lemma}\label{lem:diff_w_eq_r}
  Let $\lbrace \mcVi, \Gamma_{ij} \rbrace$ be
  a Voronoi tessellation of non-overlapping generators 
  $\lbrace g_i = (\mcBri, w_i) \rbrace$ 
  constructed from $\mcPi$ according to the difference method 
  in definition (\ref{eq:diffrule})
  with positive weights $w_i$. 
  Then the inequalities (\ref{eq:drCond0}) are satisfied for all cell pairs
  $i \neq j$ with arbitrarily small but positive cell body distance 
  $\deltij = \delta_i + \delta_j$, if and only if
  \begin{equation}\label{eq:diff_wi_eq_ri}
    \forall i: \quad w_i = r_i.
  \end{equation}
\end{lemma}
\begin{proof}: 
  The contact border equation $\mcPi = \mcPj$ evaluated at 
  the intersection point $\vx = \vSij$, see figure \ref{fig:diffrule}, 
  yields $( r_i + \delta_i )^2 - w_i^2 = ( r_j + \delta_j )^2 - w_j$.
  Together with $\delta_i + \delta_j = \deltij$ we obtain the 
  representation
  \begin{equation}\label{eq:diff_deltaij}
    \delta_{i} =
      \frac{ \deltij(\deltij+2r_j) + (w_i^2-w_j^2) - (r_i^2-r_j^2) }{
             2(r_i+r_j+\deltij) } 
    > 0.
  \end{equation}
  Since $\deltij > 0$ can be arbitrarily small for fixed $r_i, r_j$ and
  $w_i, w_j$, the condition $\delta_i > 0$ implies 
  $w_i^2 - r_i^2 \geq w_j^2 - r_j^2$. By exchanging $i$ and $j$, 
  the second condition $\delta_j > 0$ enforces the equality and 
  the existence of a joint constant $C$ with
  \begin{equation*}
    w_k^2 = r_k^2 + C \quad \text{for} \quad k = i,j.
  \end{equation*}
  Since the definition \ref{def:Vi} of a Voronoi cell 
  is independent of such an additive constant in equation (\ref{eq:diffrule}), 
  we can set $C=0$ and obtain the result of the lemma. 
\end{proof}

Thus, while a Voronoi tessellation can be formally defined 
using arbitrary subtractive weights, the constraint of 
non-overlapping cell bodies leads to the unique 
choice of weights 
$w_i = r_i$. 
Furthermore, these weights imply a simple characterization 
of the cell bodies $\mcBri = \{ \vx: \; \mcPi (\vx) < 0 \}$,
so that inequality (\ref{eq:nuclei_non_overlap}) is 
fulfilled under the assumption 
$\overline{\mcBri} \cap \overline{\mcBrj} = \emptyset$.
An illustration of a two-dimensional Voronoi tessellation with such weights
is shown in figure \ref{fig:tess}(a). 
The geometric interpretation of this choice of weights gives rise to 
the ``empty orthosphere criterion'' for a regular triangulation 
in \cite{Schaller2005,SchallerThesis,Beyer2007}, since 
the squared radius of the ``orthosphere'' equals the $\mcP$-distance 
of three or more neighboring generators from 
their common Voronoi border junction,  
consisting of red lines in figure \ref{fig:tess}(a), 
compare the analogous figure 1 \mbox{in} \cite{Beyer2007}.  

From equation (\ref{eq:diff_deltaij}) and condition 
(\ref{eq:diff_wi_eq_ri}) we obtain the dependence 
\begin{equation*}
  \delta_i =
  \frac{ \deltij ( \deltij + 2 r_{j} ) }{
         2 ( r_i + r_j + \deltij ) } \quad
  \Rightarrow \quad
  \frac{\delta_i}{\delta_j} = \frac{ \deltij + 2 r_j }{ \deltij + 2 r_i }. 
\end{equation*}
For one, if $r_i > r_j$ (as in figure \ref{fig:diffrule}), 
then $\delta_i < \delta_j$. Furthermore, for fixed $\deltij > 0$ 
and $r_j$, $\delta_i$ is monotonically decreasing in $r_i$.
This means that for growing 
cell body radius $r_i > r_j$, the distance $\delta_i$ 
between cell body and contact border $\Gamma_{ij}$ (attained 
at $\vSij$) would shrink, thus also the size of the protoplasmic 
region $\mcVi \setminus \mcBri$. 
However, such a behavior is 
contradictory to empirical observations: 
If two cells $i$ and $j$ touch each other, 
then the cell with a larger cell body should also have 
a wider cytoplasmic region along the contact border, 
see \cite{Tinkle2008} (figure 5E) and figure \ref{fig:stratum_spinosum}.
Therefore, the difference method is not appropriate and 
an alternative method is required.

\begin{figure}[h!tb]
  \centering
  \includegraphics[width=0.8\textwidth]{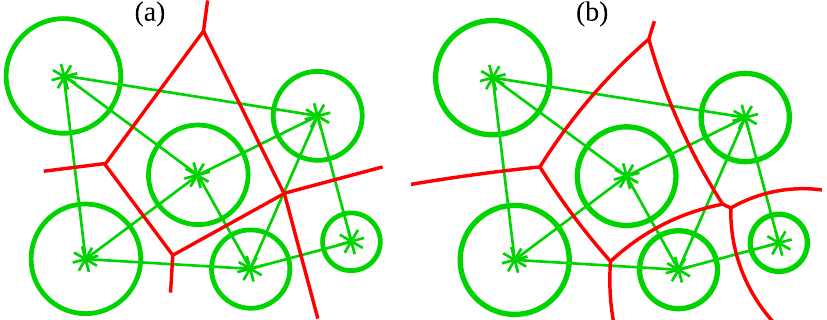}
  \caption{Comparison of two generalized Voronoi tessellations
           from the same set of generators.
           While the difference method (a) yields polygonal cells 
           the quotient method (b) yields cells with piecewise 
           spherical boundaries.
           In both cases, the weights are given by $w_i = r_i$.
           The cells are described by their cell center (green star) 
           and their body (thick green circles). The Voronoi 
           borders between cells are red lines, 
           while the neighbor relations 
           (\ie those cells that share a common border) 
           are indicated by thin green lines connecting their centers.
          }\label{fig:tess}
\end{figure}

\subsection{Quotient method}

In the previous section it was found that, with subtractive
weights in the $\mcP$-distance of the generalized Voronoi 
tessellation, the emerging cell contact border are planar 
surfaces. In contrast, if one divides the Euclidean distance by 
weights, then the cell contacts are spherical  
with the generalized square distance function defined as 
\begin{equation}\label{eq:ratrule}
  \mcPi ( \vx ) = \frac{ ( \vx - \vxi )^2 }{ w_i^2 }.
\end{equation}
Having its roots in computational geometry 
(see \cite{Aurenhammer1983} and references therein), 
this method was introduced as a model for attraction 
domains of restaurants \cite{Ash1986} more than 20 years ago. 
To our knowledge, it so far has not 
been used for physical or biological applications.

For simplicity of calculation, let the midpoint $\vxij := 0$ be 
the origin of the coordinate system, while 
$\udij$ remains the oriented cell-cell axis
(see figure \ref{fig:ratrule}). 
\begin{figure}[htb]
  \centering
  \includegraphics[width=0.6\textwidth]{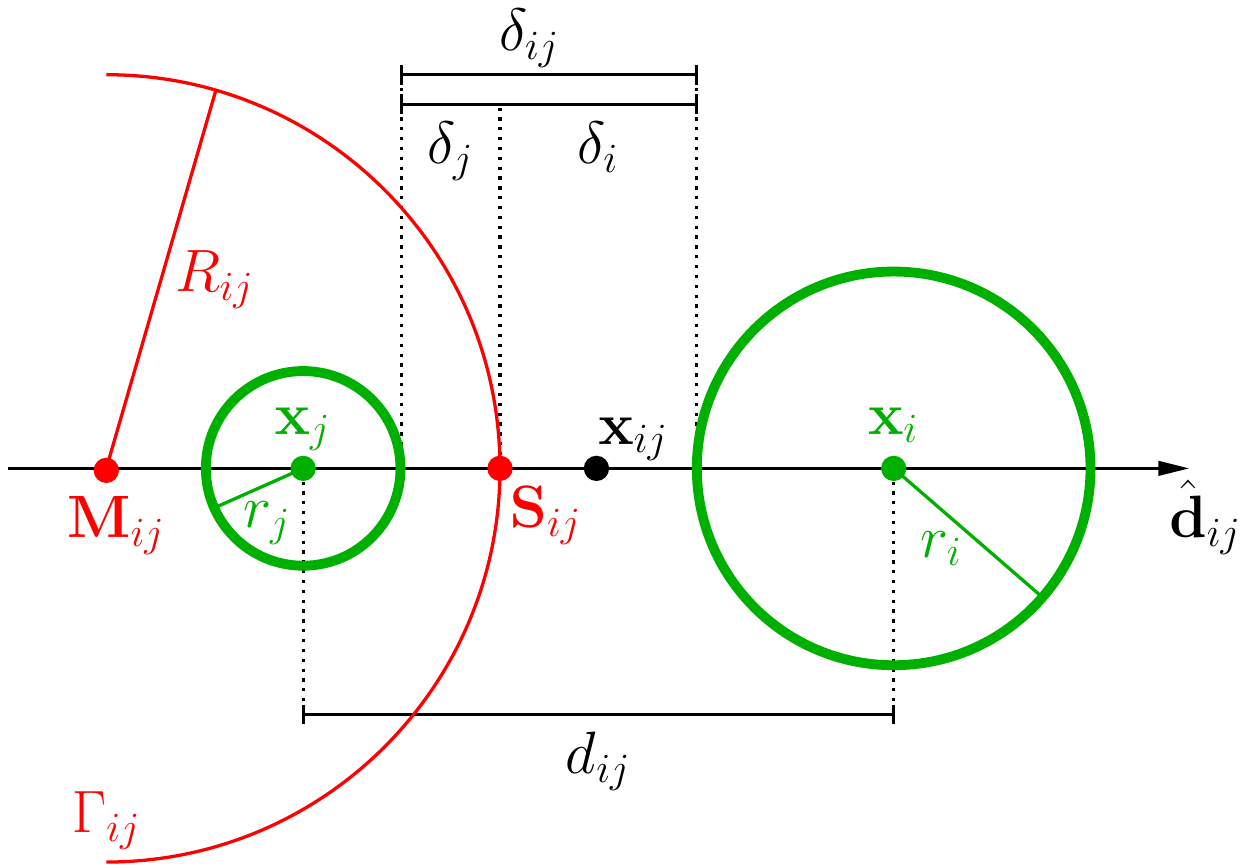}
    \caption{Geometry of cell pair and its pair contact 
             border using the quotient method. 
             In contrast to the difference method, this contact 
             border $\Gamma_{ij}$ is a sphere around 
             $\vMij$ with radius $\Rij$. Its two-dimensional 
             section is drawn in a red line.
            }\label{fig:ratrule}
\end{figure}
Starting from $\mcPi(\vx) = \mcPj(\vx)$ with $w_i \neq w_j$ 
one arrives at the equivalent condition
\begin{equation}\label{eq:ratrule_boundary}
  \left( \vx - \vMij \right)^2 = \Rij^2
\end{equation}
for the point $\vx$ to be on the border $\Gamma_{ij}$. 
Clearly, equation (\ref{eq:ratrule_boundary}) 
describes an $n$-sphere around
\begin{equation}\label{eq:ratrule_MR} 
  \vMij = - \frac{ w_i^2 + w_j^2 }{ w_i^2 - w_j^2 } \vxi \qquad
  \text{with radius} \qquad
  \Rij = \frac{ w_i w_j }{ (w_i^2 - w_j^2) } \dij,
\end{equation}
where $\dij = | \vxi - \vxj | = 2 |\vxi|$, 
resulting in a so called \emph{circular Voronoi tessellation}. 
Assuming $r_i > r_j$ (as in figure \ref{fig:ratrule}), also the 
weights fulfill $w_i > w_j$ according to our monotonicity assumption 
on $w_i = w(r_i)$. Thus, from equation (\ref{eq:ratrule_MR}) the 
center $\vMij$ of the hypersphere $\Gamma_{ij}$ is always situated 
on the side of the cell $j$ with the smaller radius $r_j$.
The contact sphere intersects 
the cell center connection segment $\overline{ \vxj, \vxi }$
at a unique contact point determined by 
\begin{equation}\label{eq:ratrule_Sij}
  \vSij \cdot \udij = - \frac{ w_i - w_j }{ w_i + w_j } \cdot \frac{\dij}{2}.
\end{equation}
Similar as for the difference method (see figure \ref{fig:diffrule}), 
$\vSij$ is situated on the side 
of the smaller cell body from the mid point $\vxij = 0$. 
Thus, the contact sphere contains the body of the cell with 
smaller weight, as indicated in figure \ref{fig:ratrule}. 
Once both weights are equal, equation (\ref{eq:pi_eq_pj}) 
and thereby equation (\ref{eq:ratrule_Sij}) simplify to 
\begin{equation}\label{eq:ratrule_wi_eq_wj}
  \vx \cdot \udij = \vSij \cdot \udij = 0 \qquad \text{for} \; \, w_i = w_j,
\end{equation}
revealing $\Gamma_{ij}$ as the classical Voronoi bisector 
line without weights.

In analogy to Lemma \ref{lem:diff_w_eq_r} for the difference method, 
we obtain the same unique specification of weight functions here as well:
\begin{lemma}\label{lem:rat_w_eq_r}
  Let $\{ \mcVi, \Gamma_{ij} \}$ be a Voronoi tessellation 
  of non-overlapping generators $\{ g_i = (\mcBri, w_i) \}$ 
  constructed from $\mcPi$ according to the quotient method 
  in definition (\ref{eq:ratrule}) 
  with positive weights $w_i$.
  Then the inequalities (\ref{eq:drCond0}) are satisfied 
  for all cell pairs $i \neq j$ 
  with arbitrarily small but positive 
  cell body distance $\deltij = \delta_i + \delta_j$ 
  if and only if
  \begin{equation}\label{eq:rat_wi_eq_ri}
    \forall i: \quad w_i = r_i.
  \end{equation}
\end{lemma}
\begin{proof}:
  The contact border equation $\mcPi = \mcPj$ evaluated 
  for the point $\vx = \vSij$, see figure \ref{fig:ratrule}, 
  yields $(r_i +  \delta_i)/w_i = (r_j + \delta_j)/w_j$. Together
  with $\delta_i + \delta_j = \deltij$ we obtain the representation 
  \begin{equation}\label{eq:ratrule_deltaij}
    \delta_{i} = \deltij \frac{w_{i}}{ w_i + w_j } + 
                 \frac{ r_j w_i - r_i w_j }{ w_i + w_j } > 0. 
  \end{equation}
  Since $\deltij > 0$ can be arbitrarily small for fixed $r_i, r_j$ 
  and $w_i, w_j$, the condition $\delta_i > 0$ implies 
  $r_i w_j \geq r_j w_i$. By exchanging $i$ and $j$, 
  the second condition $\delta_j > 0 $ enforces equality and 
  the existence of a joint positive constant with
  \begin{equation*}
    w_k = r_k \cdot C \quad \text{for} \quad k = i,j.
  \end{equation*}
  Since the definition \ref{def:Vi} of a Voronoi cell is 
  independent of such a multiplicative constant in equation (\ref{eq:ratrule}), 
  we can set $C=1$ and obtain the result of the lemma. 
\end{proof}
Thus, further on we can choose the weights $w_i = r_i$ when 
using the quotient method. Then the cell bodies are characterized 
as $\mcBri = \{ \vx: \, \mcPi (\vx) < 1 \}$, so that inequality 
(\ref{eq:nuclei_non_overlap}) is fulfilled under the assumption 
$\overline{\mcBri} \cap \overline{\mcBrj} = \emptyset$.
The emerging circular Voronoi tessellation 
is illustrated in figure \ref{fig:tess}(b).
Rewriting equation (\ref{eq:ratrule_deltaij}) for i and j yields
\begin{equation*}
  \delta_{i,j} = \deltij \frac{ r_{i,j} }{ r_i + r_j },
\end{equation*}
and thus the regular partition property
\begin{equation}\label{eq:ratrule_part}
  \frac{ \delta_j }{ \delta_i } = \frac{ r_j }{ r_i },
\end{equation}
meaning that the partition of the distance between cell bodies 
$\mcBri$ and $\mcBrj$ by the contact arc
is proportional to the ratio of body radii. 
Most importantly, in contrast to the difference method, 
the $\mcP$-distance in equation (\ref{eq:ratrule}) 
ensures that $\delta_i$ monotonically increases with $r_i$, while 
$\delta_j = \deltij - \delta_i$ decreases, if $r_j > 0$ and $\deltij > 0$ 
are held fixed. As a consequence, the cell $i$ grows with increasing 
$r_i$. This property can also be observed for in-vivo cell monolayers, 
\begin{figure}[bthp]
  \centering
  \includegraphics[width=0.9\textwidth]{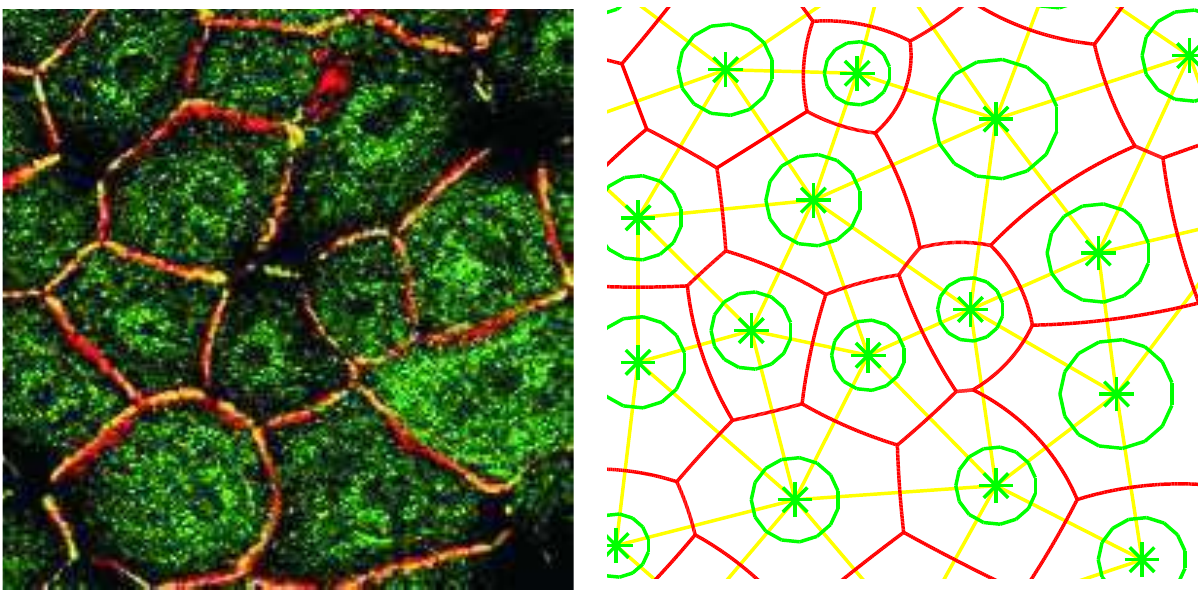}
  \caption{Typical microscopic picture of an epithelial monolayer (left)
           cultured of human keratinocytes, 
           cell nuclei (green, black) and cell-cell contacts 
           (yellow, red) are visualized by suitable staining
           (reproduced from \cite{Marie2003}).
           In the simulated cell tissue (right) the sizes and 
           positions of cell bodies (green) have been roughly 
           adapted to deliver the observed contact arcs (red).
          }\label{fig:cell_model}
\end{figure}
see figures \ref{fig:stratum_spinosum}, \ref{fig:cell_model} and 
figure 5E \mbox{in} \cite{Tinkle2008}.
Thus, for the further discussions in this paper, 
the quotient method will be used exclusively to define 
Voronoi cells. Yet the tessellation is still unbounded due to 
infiniteness of marginal cells.

\subsection{Closure of the Voronoi tessellation}

In order to avoid infinitely extended cells at the tissue margin, 
the initial definition \ref{def:Vi} of a Voronoi cell has to be modified.
To this end, the method of finite closure for the difference 
method, as used by Drasdo and coworkers \cite{Drasdo1995,Galle2005},
is extended to be generally applicable.
\begin{definition}\label{def:closed_Vi}
  Let $\mcPm \in \mathbb{R}^+$.
  The finite Voronoi cell of a generator  $( \mcBri, w_i > 0 )$ 
  is defined as
  \begin{equation}\label{eq:cVNG}
    \mathcal{V}_i = \left\{
        \mathbf{x} \in \mathbb{R}^n :
          \mcPi (\vx) < \min{\left( \mcPj (\vx), \mcPm \right)} 
        \quad \forall j \neq i
      \right\}.
  \end{equation}
  The exterior boundary closing a marginal $\mathcal{V}_i$ is
  \begin{equation}\label{eq:closed_Gammaij}
    \Gamma_{i0} = \left\{ \vx \in \mathbb{R}^n : \mcPi (\vx) = \mcPm \right\}
        \setminus \bigcup_{j \neq i} \mathcal{V}_j,
  \end{equation}
  and the total boundary of the Voronoi neighborhood around $\vxi$ is
  \begin{equation*}
    \partial \mathcal{V}_i = \Gamma_{i0} \cup \bigcup_{j \neq i} \Gamma_{ij}, 
  \end{equation*}
  where now the contact border between cell $i$ and $j$ is given by
  \begin{equation}
    \Gamma_{ij} = \Bigl\{ \vx \in \mathbb{R}^n: \; 
        \mcPi (\vx) = \mcPj (\vx) \leq 
        \min \bigl( \mathcal{P}_k (\vx), \mcPm \bigr) \quad 
        \forall k \neq i,j \Bigr\}.
  \end{equation}
\end{definition}
For any choice of $\mcPm > 0$, the Voronoi tessellation generated 
from a finite set of cell bodies $\{ ( \mcBri, w_i ) \}$ 
comprises a bounded region of the whole space, 
representing a cell tissue of finite extension.
In figure \ref{fig:vor_closure} we depict the model representation of an 
\begin{figure}[htbp]
  \centering
  \includegraphics[width=0.6\textwidth]{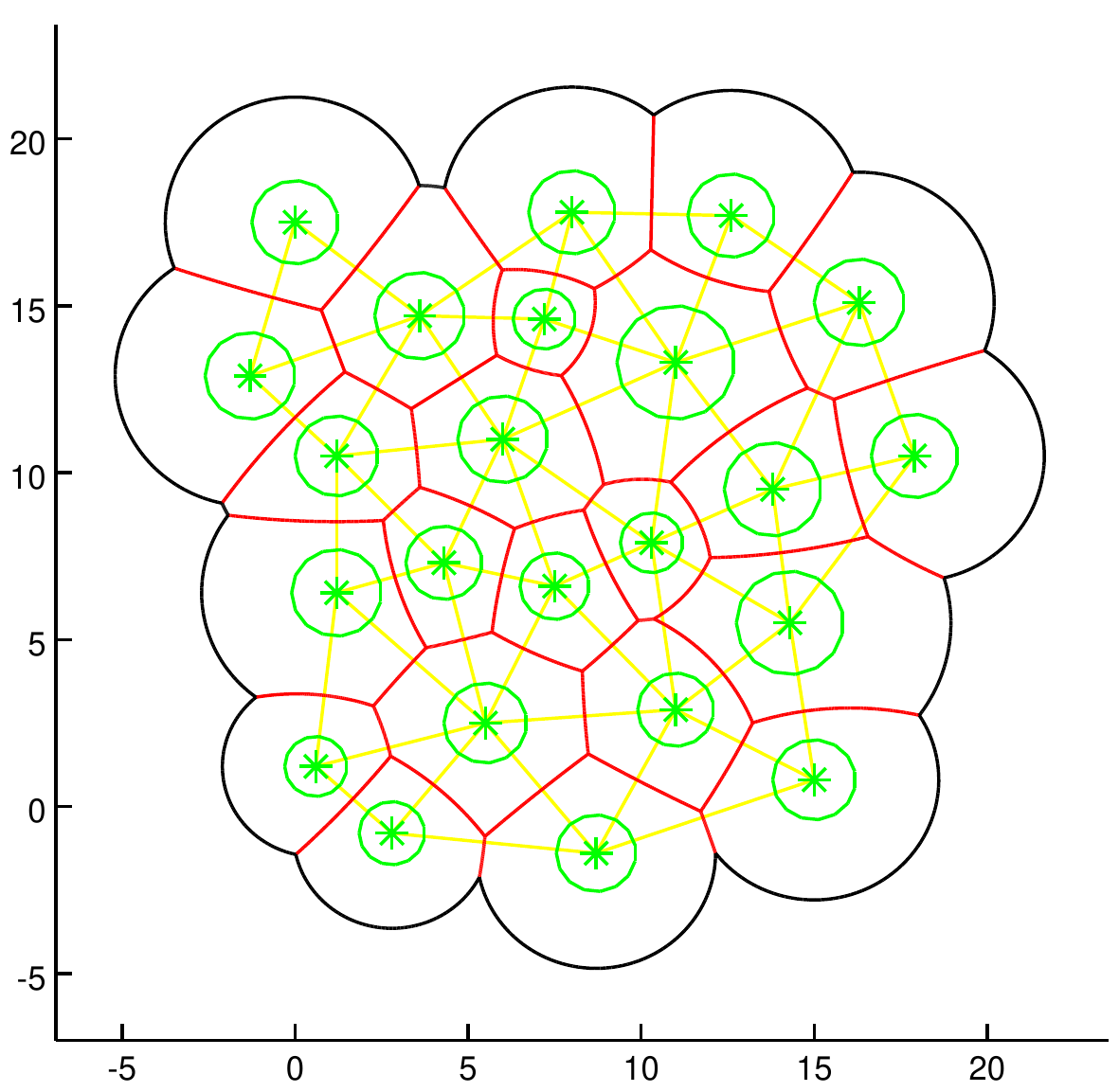}
    \caption{Model representation of a bounded cell monolayer using the 
             quotient method with $\mcPm$-cutoff. Axis 
             tics are in units of $\mu \mathrm m$.
            }\label{fig:vor_closure}
\end{figure}
\emph{in vivo} cell monolayer (see figure \ref{fig:cell_model}) 
generated by the quotient method 
using definition (\ref{eq:ratrule}) with $w_i = r_i$. 
The exterior boundaries $\Gamma_{i0}$ of the Voronoi neighborhoods 
for marginal cells are circular arcs drawn as black lines.
Apparently the size of a cell $i$
is influenced by the two parameters $r_i$ and $\mcPm$.
While $\mcPm$ regulates the overall cell size 
by globally scaling each $r_i$, the ratios $r_i / r_j$
determine the partition of space between each cell pair $i,j$ 
by specifying the actual position of $\Gamma_{ij}$. 
We remark that both $\mcPm$ and $\{ r_i \}$ are accessible to experimental 
determination using image analysis tools, see figure \ref{fig:cell_model}
and especially figure \ref{fig:stratum_spinosum}(c).

The necessary condition $\mcPi (\vx) < \mcPm$ for a 
point $\vx$ to be within cell $i$ defines a ball $\mcBRif$ 
around $\vxi$, which will be called free ball further on.
It has the squared radius
$\Rif^2 = w_i^2 + \mcPm$ in the difference method and 
$\Rif^2 = w_i^2 \cdot \mcPm$ in the quotient method. 
As illustrated in figure \ref{fig:pmax}, 
\begin{figure}
  \centering
  \includegraphics[width=0.8\textwidth]{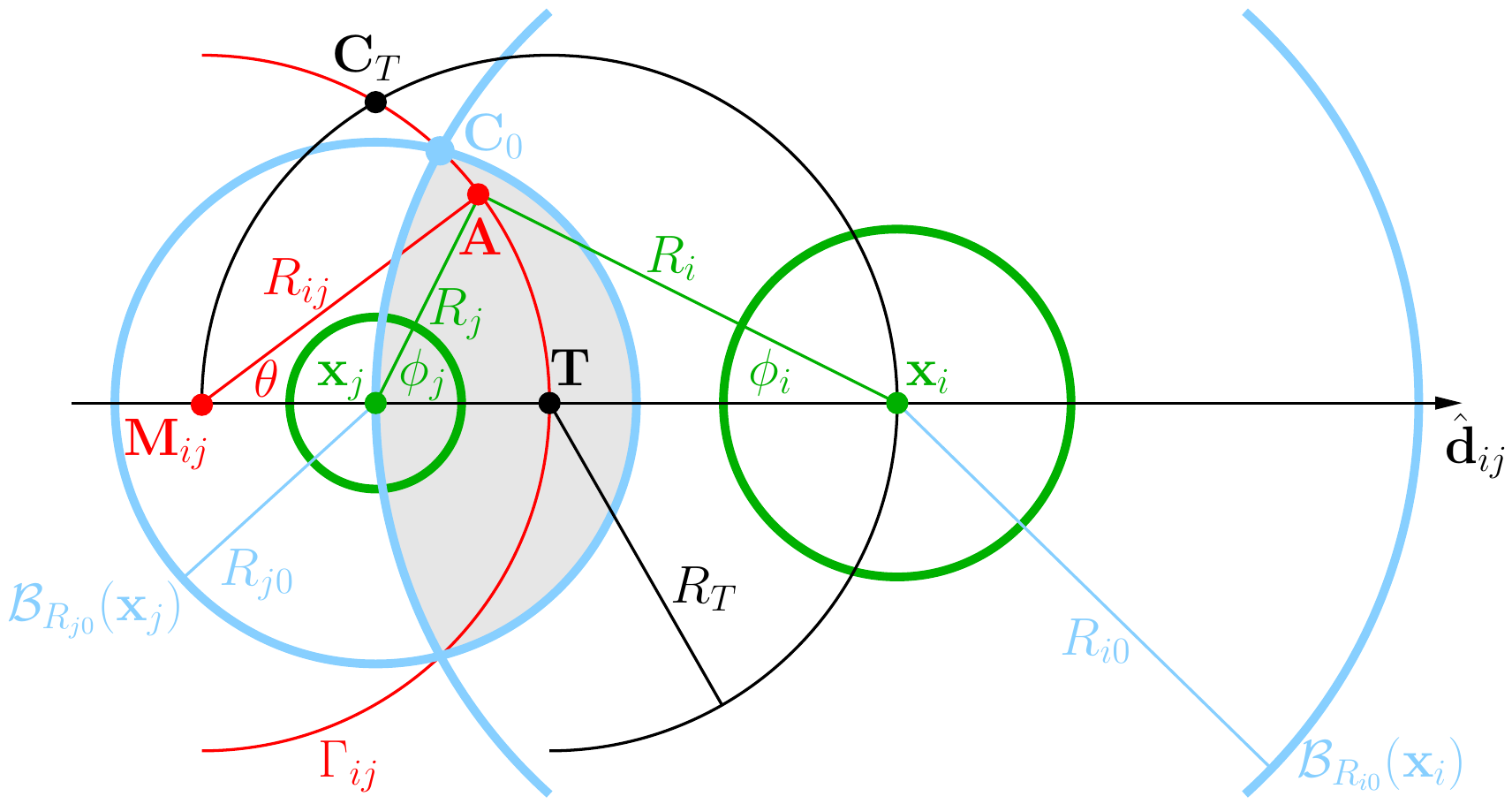}
  \caption{Maximal $\mcP$-distance $\mcPm$ and its effect on the 
           neighbor relation of two cells $i,j$. The cells $i,j$ can only 
           share a common $\Gamma_{ij}$ and thus be neighbors, if 
           there is a non-empty overlap region (shaded) given by the 
           intersection of their free balls $\mcBRif, \mcBRjf$ 
           (blue outer circles). The other quantities are explained 
           in the text.
          }\label{fig:pmax}
\end{figure}
this has the effect that two cells $i,j$ can only be neighbors, 
if their free balls $\mcBRif, \mcBRjf$ overlap.
In case of no contact, $\mcVi = \mcBRif$ represents an 
isolated 
spherical cell, which 
clearly has to include its cell body $\mcBri$. Thus, the condition 
\begin{equation}\label{eq:Rif}
  r_i^2 < \Rif^2 = \left\{ \begin{aligned}
      & r_i^2 + \mcPm \; & \text{(difference method)} \\
      & r_i^2 \cdot \mcPm \; & \text{(quotient method)}
  \end{aligned} \right.
\end{equation}
is imposed, meaning that $\mcPm$ has to be chosen so that 
\begin{equation}\label{eq:pmax_min}
  \begin{aligned}
    & \mcPm > 0 \quad \text{(difference method)}& \\
    & \mcPm > 1 \quad   \text{(quotient method)}.&
  \end{aligned}
\end{equation}
In this way, the larger $\mcPm$, the larger is the radius of isolated 
cells relative to their cell body. 

Within the difference method, this $\mcPm$-closure is 
straight-forward because the planar cell contacts lead to 
starlike (even convex) Voronoi cells $\mcVi$.
Recall the definition of starlikeness with respect to the center $\vxi$: 
$\forall \vx \in \mcVi$ also $\overline{\vxi \vx} \subset \mcVi$. 
In fact, a $\mcPm$-closed generalized Voronoi tessellation 
has been applied to epithelial tissue modeling  
by prescribing $\mcBRif$ for each cell \cite{Galle2005}. 
However, within the quotient method the situation is more complicated. 
In analogy to the difference method it is reasonable to require 
that the Voronoi cells $\mcVi$ are star-like domains with respect to $\vxi$.
In order to ensure starlike cells within the quotient method, 
$\mcPm$ must not be chosen too large. 
Consider the cell pair as sketched in figure \ref{fig:pmax}.
The straight line connecting $\vxi$ and $\vCT$ is a tangent to 
$\Gamma_{ij}$. Thus it is clear from the geometry 
that both $\mcVi$ and $\mcVj$ are star-like domains with 
respect to $\vxi, \vxj$, if their corresponding 
free balls $\mcBRif, \mcBRjf$ do not
extend beyond the point $\vCT$. 
Before we proceed, we introduce the \emph{cell size homogeneity} quotient
\begin{equation}
  \qqq = \min_{i,j} \frac{ r_i + r_j }{| r_i - r_j |} =
         \frac{ r_\text{max} + r_\text{min} }{ r_\text{max} - r_\text{min} },
\end{equation}
where the last equality follows from monotonicity arguments.
Therefore, $Q = Q( \{ r_i : i = 1 \dots N \} )$ is a measure 
of the uniformity of cell sizes within a tissue, with $Q = \infty$ 
for equal $r_i$  and $Q \approx 1$ for $r_\text{max} >> r_\text{min}$.

\begin{proposition}[Starlike cells]\label{prop:starlike}
  For a finite Voronoi tessellation generated from 
  non-overlapping $\{ \mcBri \}$ by 
  using the quotient method in definition (\ref{eq:ratrule}) 
  with weights $w_i = r_i$, 
  the resulting Voronoi cells $\mcVi$ are starlike with respect to $\vxi$, 
  if the maximal $\mcP$-distance $\mcPm$ fulfills the homogeneity constraint 
  \begin{equation}\label{eq:pmax_max1}
    1 < \mcPm \leq \qqq.
  \end{equation}
\end{proposition}
This condition on the tissue properties will be crucial later on 
and guarantees that each actin fiber bundle emanating 
radially from $\partial \mcBri$ intersects the boundary $\partial \mcVi$ 
only once, see figure \ref{fig:Density}.
\begin{proof}:
  From fundamental trigonometric relations follows 
  the angle $\angle ( \vT, \vxj, \vCT )$, namely
  \begin{equation}\label{eq:phis_max}
    \phi_j^T = \frac{ \pi }{ 2 }, 
  \end{equation}
  and geometric similarity of the triangles 
  $\triangle( \vMij, \vxi, \vCT ), \triangle( \vxi, \vCT, \vxj )$.
  Thus, with $r_j < r_i$, we have for point $\vA = \vCT$
  \begin{gather}
    \cos \thm = \frac{ r_j }{ r_i } 
                                                 \label{eq:thetamax} \\
    \RjT = \frac{ r_j }{ \sqrt{| r_i^2 - r_j^2 |} } \cdot \dij, \qquad 
    \RiT = \frac{ r_i }{ \sqrt{| r_i^2 - r_j^2 |} } \cdot \dij.
                                                 \label{eq:RiRjmax}
  \end{gather}
  With the last two equations, the maximal distances of a 
  point $\vA$ on $\Gamma_{ij}$ from the cell centers 
  have been identified for each cell pair. 
  Starlikeness of $\mcVi$ is equivalent to the condition 
  $\Rif^2 \leq \RiT^2$, where $\Rif^2 = \mcPm r_i^2$, so 
  that $\mcPm \leq \dij^2 / | r_i^2 - r_j^2 |$, 
  which can be fulfilled by requiring $\mcPm \leq \qqq$,
  since $\forall i,j: \; (r_i + r_j)^2 \leq \dij^2$. 
  With the condition (\ref{eq:pmax_min}) the assertion 
  follows. 
\end{proof}
In particular, starlikeness prohibits engulfment of one cell by the other, 
so that $\mcBRif$ may not contain $\mcBRjf$ completely for $r_i > r_j$. 
Note that within sufficiently large tissues, the smallest and 
biggest cell will usually not be in contact, which 
relaxes inequality (\ref{eq:pmax_max1}) into the condition:
\begin{equation}\label{eq:pmax_max2}
  1 < \mcPm \leq 
    \min_{ \text{\rm neighbors} \; i,j } \frac{ r_i + r_j }{ | r_i - r_j | } 
  =: \qnb.
\end{equation}

In such a circular Voronoi diagram there may be $\mathcal O(N^2)$ 
cell-cell contacts and vertices, 
in particular for low cell size homogeneity quotients $Q$.
The supplementary material contains the extensively commented 
\texttt{GNU octave} routine 
\href{http://www.theobio.uni-bonn.de/people/mab/sup01/index.html}{\texttt{mwvoro.m}}
(Matlab compatible) used to create and visualize a 
circular closed Voronoi tessellation, together with 
configuration and plotting facilities. 
The partition of space into distinct cells and their 
neighborhood relations has an algorithmic complexity of 
$\mathcal O(N^2)$, $N$ being the number of cells. 
Due to the extension of closing marginal cells by $\mcPm$-arcs, 
we do not follow the method proposed by Aurenhammer and 
Edelsbrunner \cite{Aurenhammer1983}. 
In particular, we do not use polyedral ``cell complexes'' 
in an inverted three-dimensional embedding of the Voronoi generators.
Nevertheless we retain the same optimal algorithmic complexity.

\section{Cell shape and dynamics}\label{sec:forces}

When a single cell is placed on a two-dimensional and 
adhesive substratum, it usually spreads into all directions 
attaining a circular shape like a fried egg,
as can be observed in figure \ref{fig:kcyt}.
\begin{figure}
  \centering
  \includegraphics[width=0.6\textwidth]{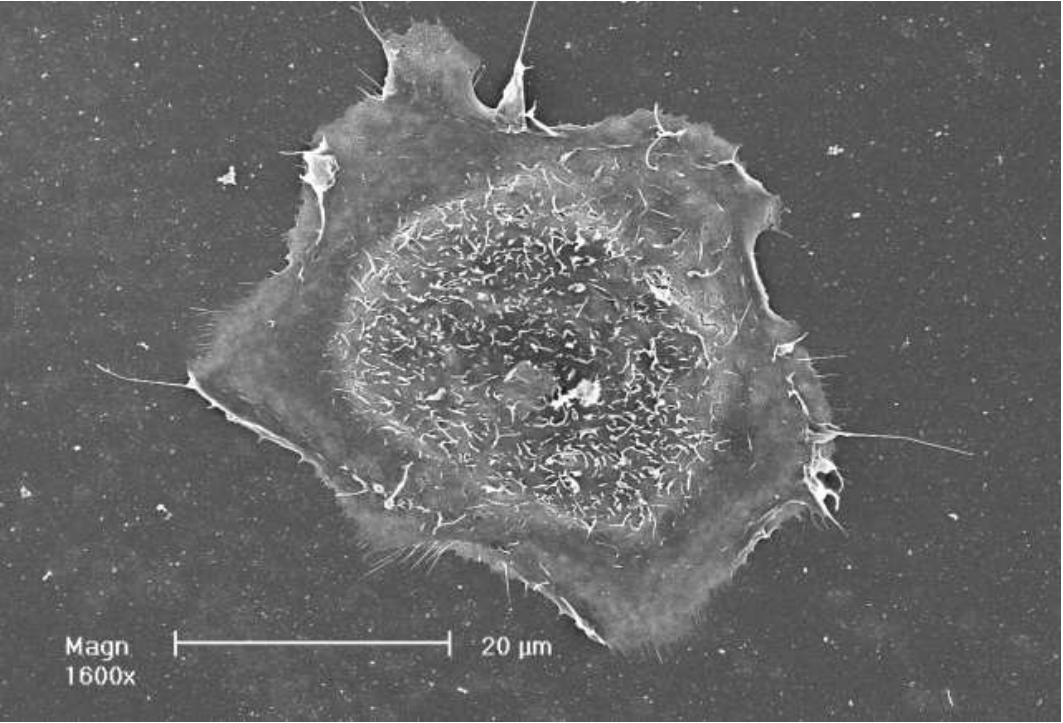}
    \caption{Electron microscopic picture of a keratinocyte
             (courtesy of Gregor Wenzel).
             There are two distinct regions within the cell. An inner, 
             more dense region with ruffled membrane structures appearing 
             in white (cell body), and an outer, flat region with 
             larger ruffles or filopodia near the cell margin (lamella),
             which in the case of no contact with other cells forms a
             ring of more or less constant diameter around the cell body. 
             See also figure \ref{fig:stratum_spinosum}.
    }\label{fig:kcyt}
\end{figure}
Thereby, the exterior visco-elastic lamella 
along the free boundary $\Gamma_{i0}$, which consists of parts of $\mcBRif$, 
supports the protruding and retracting cell edge.
This smooth flat region contains a network of dynamic actin filaments 
situated around the inner, 
almost solid cell body $\mcBri$, see \eg \cite{Alberts,Sivaramakrishnan2008}. 
The maximal spreading radius $\Rif = \sqrt{ \mcPm } \cdot r_i$ 
is determined by the strength of adhesion to the substrate 
and the equilibrium between protrusive activity at the cell 
periphery $\Gamma_{i0}$ and the contractile retrograde actin flow
\cite{KuuselaAlt2009}.
Assuming that for given adhesiveness 
the averaged local cytoskeletal network 
volume fraction ($\theta$ in \cite{KuuselaAlt2009}) 
in the lamellae has a certain value $q_0$ independent of cell size, 
the weighting of $\Rif$ 
proportional to the cell body radius follows.

Usually, the irregular activity of living cells at their 
lamella edge leads to a curled cell boundary, see figures 
\ref{fig:stratum_spinosum}(c) and \ref{fig:kcyt}.
Neglecting fluctuations on short time $\mathcal{O}(10s)$ and 
length scales $\mathcal{O}( 1\mu \mathrm{m} )$, 
the free portions of the cell edge $\Gamma_{i0}$
are approximately taken as circular in this model.
The actual cell edge fluctuates within the vicinity of the smooth arcs, 
representing the averaged position of the plasma membrane, 
and will later on be taken as the source of 
stochastic perturbation forces, see section \ref{sec:floc}.
Thus, in our model the active lamella 
region of a single free cell is approximately ring shaped and 
has a width of 
\begin{equation}\label{eq:lamwid}
  \deltif = ( \sqrt{\mcPm} - 1 ) r_i.
\end{equation}

Once two epithelial cells $i,j$ come close enough to interact, 
the two adjoining lamellae compete for the occupation 
of the region in between them.  
Eventually they form a contact border, which exhibits microscopic 
fluctuations due to local plasma membrane flickering. 
Yet it approximately attains the shape of a circular arc,  
whereby small gaps between the two cell membranes are neglected.
This experimental fact (see \eg corresponding figures in
\cite{Young2000,Marie2003,Tinkle2008} and
\ref{fig:cell_model}, \ref{fig:stratum_spinosum} in this article) 
is well represented by the Voronoi border $\Gamma_{ij}$ resulting 
from the quotient method defined by equation (\ref{eq:ratrule}). 
In this way, within our tissue model, 
the cell boundaries are merely composed of piecewise circular arcs, 
and the cell bodies are not necessarily located in the middle of the cells. 
By suitable choice of $w_i$ (Lemma \ref{lem:rat_w_eq_r})
there is always some lamella region separating the cell body from 
the neighbor cell ($\delta_i > 0$, also \cf figure \ref{fig:ratrule}).

\subsection{Interaction forces between cells}

The cytoskeleton with its network of filaments 
often features bundled structures, 
which are commonly visible as so-called stress fibers, 
emanating from the cell body or nucleus in radial direction.
According to \cite{Alberts}, bundles of filamentous actin 
attach to transmembrane complexes called adherens junctions, 
which are made from \eg catenins on the cytosolic side and cadherins 
at the exterior of the cell. 
Furthermore, intermediate filaments such as the rope-like keratin 
tie in with rivet-like desmosomes at the cell membrane. 
By connecting neighboring cells, these structures stiffen 
and strengthen the tissue coherence. 
For example, in certain epithelia, cadherin-catenin adherens junctions 
comprise a whole transcellular adhesion belt.

Inspired by this observation, it is assumed 
that the attractive force between two cell bodies $\mcBri, \mcBrj$ 
is transduced by radial filament structures
extending towards the cell boundaries. 
Thereby, the filaments of one cell 
connect to those of the other 
and form pairs along the contact border $\Gamma_{ij}$. 
Thus, the intercellular adherens junctions emerge 
according to the respective filament densities 
as emanating from cell $i$ and $j$, respectively 
(see figure \ref{fig:Adhesion}). 
\begin{figure}[htb]
  \centering
  \includegraphics[width=0.6\textwidth]{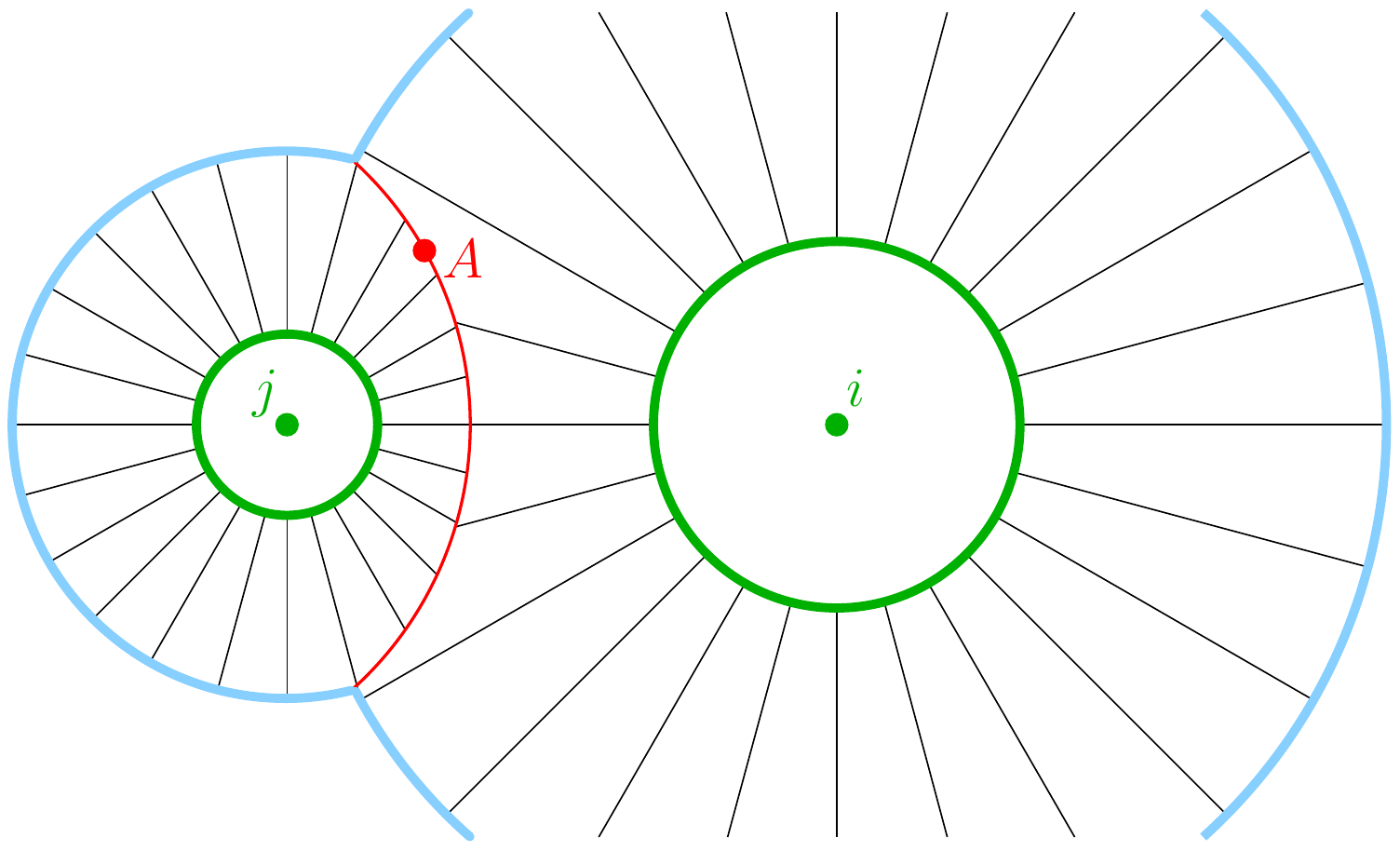}
    \caption{Pairing of filaments from one cell to the other cell.
             For the definition of symbols and angles to describe 
             the geometry of pairing filaments of a  
             neighboring cell pair see figure \ref{fig:pmax}.
            }\label{fig:Adhesion}\label{fig:Density}
\end{figure}
Furthermore, the connecting cell-cell junctions are not fixed but undergo 
dissociation, diffusion, and renewed association. 
Motivated by protein (\eg cadherin) 
diffusion properties in membranes \cite{Gambin2006,Brevier2008}, 
this process is considered to be fast (seconds) compared to 
the slower time scale (several minutes) 
of cell deformation and translocation.
In this way, 
pair formation of cross-attachments between filament bundles 
from both cell bodies can be regarded as a pseudo-stationary stochastic 
process \cite{Evans1997}. 
In order to compute the interaction force between two cells, 
one needs a suitable expression for the density of pairing filaments 
$\rho (\theta)$ at the border of the cells $i,j$.

\subsection{Filament pair density at cell-cell contacts}\label{sec:filapair}

Consider the cell pair as illustrated in figures \ref{fig:pmax} and 
\ref{fig:Density} with cell body radii $r_i > r_j$ and the 
distinguished Voronoi weights as in equation (\ref{eq:rat_wi_eq_ri}).
Let $\theta$ parameterize the contact arc $\Gamma_{ij}$ 
given by  $\vMij, \Rij$, and let $\vA$ be the corresponding 
point upon that arc.
Starting from the surface of the cell bodies $\mcBri$ and $\mcBrj$, 
filaments extend in radial direction under angles 
$\phi_i (\theta)$ and $\phi_j (\theta)$, respectively, 
to eventually meet at $\vA$. 
Furthermore, let
$R_i$ and $R_j$ denote the distances between the cell centers and $\vA$. 
The density of filaments is assumed to be constant on the 
surface of each cell body, given by a universal value $\rhon > 0$. 
In order to construct the pairing density of filaments $\rho (\theta)$ 
along the contact surface,
these cell body surface densities are mapped onto $\Gamma_{ij}$ 
by equating the corresponding surface elements
\begin{equation}\label{eq:map_rho}
    \rho_i (\theta) \Rij \ddd \theta = \rhon r_i \ddd \phi_i, \qquad
    \rho_j (\theta) \Rij \ddd \theta = \rhon r_j \ddd \phi_j.
\end{equation}
With $\vA = \vA (\theta) \in \Gamma_{ij}$, 
(\cf figures \ref{fig:pmax} and \ref{fig:Density}), it holds
\begin{gather}
  \Rij \sin{\theta} =  R_i \sin{\phi_i} = R_j \sin{\phi_j}
      \label{eq:fdsin} \\
  \Rij \cos{\theta} = |\vxj - \vMij| + R_j \cos{\phi_j}.
      \label{eq:fdcos}
\end{gather}
The defining condition for the contact border in equation 
(\ref{eq:pi_eq_pj}) can be written as
\begin{equation}\label{eq:fdReta}
  R_j = \eta R_i \qquad \text{with} \; 
  \eta = \frac{w_j}{w_i} = \frac{r_j}{r_i} < 1.
\end{equation}
Thus, from equation (\ref{eq:fdsin}) we obtain the simple relation
\begin{equation}\label{eq:phieta}
  \sin{\phi_i} = \eta \cdot \sin{\phi_j}
\end{equation}
between the two angles $\phi_i (\theta)$ and $\phi_j (\theta)$, so that
differentiation with respect to $\theta$ yields the proportionality
\begin{equation}\label{eq:dphidth}
  \frac{ \ddd \phi_i }{ \ddd \theta } =
  \eta \cdot \frac{ \cos{\phi_j} }{ \cos{\phi_i} } \cdot
    \frac{ \ddd \phi_j }{ \ddd \theta } =
  \frac{ \eta }{ \kappa_\eta ( \phi_j ) } \cdot
    \frac{ \ddd \phi_j }{ \ddd \theta },
\end{equation}
where $\kappa_\eta ( \phi_j ) = 
\sqrt{ 1 + (1-\eta^2) \cdot \tan^2 \phi_j }$.
Moreover, by solving equation (\ref{eq:fdsin}) for 
$R_j$ in terms of $\Rij$, inserting it into equation (\ref{eq:fdcos}), 
and using the relations (\ref{eq:ratrule_MR}) 
we get an explicit expression for $\tan{\phi_j}$ in terms of $\theta$
\begin{equation}\label{eq:tanphij}
  \tan{\phi_j} = \frac{ \sin{\theta} }{ \cos{\theta} - \eta },
\end{equation}
which holds for all $|\theta| < \thm$, with $\cos \thm = \eta$, 
or equivalently, $|\phi_j| < \pi/2$, 
see equations (\ref{eq:phis_max},\ref{eq:thetamax}).
Finally, by differentiation of equation (\ref{eq:tanphij}) 
with respect to $\theta$ we obtain 
\begin{equation}\label{eq:dphjdth}
  \frac{ \ddd \phi_j }{ \ddd \theta } = 
  \frac{ \tan{\phi_j} }{ 1 + \tan^2 \phi_j } \cdot
    \biggl( \tan{\phi_j} + \frac{1}{ \tan{\theta} } \biggr) =
  \frac{ 1 - \eta \cos{\theta} }{ 1 - 2 \eta \cos{\theta} + \eta^2 } >
  0.
\end{equation}

It is assumed, that the pairing density function $\rho (\theta)$, 
depending on $\rho_i (\theta)$ and $\rho_j (\theta)$, is 
even in $\theta$, maximal at $\theta=0$, and strictly monotonically 
decreasing for increasing $|\theta|$. Here, two exemplary models 
to specify such a density function $\rho (\theta)$ are discussed:

\paragraph{Model 1: Minimal density pairing.}
If locally one cell has less filaments binding to $\Gamma_{ij}$ than 
the other, then there will be a pairing match for all of its filaments.
Thus, the local density of pairs on $\Gamma_{ij}$ will equal the 
lower filament density:
\begin{equation}\label{eq:rho_theta_1}
  \rho(\theta) = \min \Bigl( \rho_i (\theta), \rho_j (\theta) \Bigr) =
                 \frac{ \tilde{\rho} }{\Rij} 
                 \min \biggl( 
                   r_i \frac{ \ddd \phi_i }{ \ddd \theta }, 
                   r_j \frac{ \ddd \phi_j }{ \ddd \theta }
                 \biggr),
\end{equation}
where we used the identities (\ref{eq:map_rho}).
With $\kappa_\eta > 0$ we conclude from equation (\ref{eq:dphidth}) 
that $\rho (\theta) = \rho_i (\theta) \leq \rho_j (\theta)$.
Therefore, an explicit representation of $\rho$ 
in terms of $\phi_j$ and its derivative is
\begin{equation*} 
  \rho_\text{min} (\theta) = \rho_i (\theta) = 
  \tilde{\rho} \cdot
    \frac{ r_j }{ \Rij \cdot \kappa_\eta ( \phi_j ) } \cdot
    \frac{ \ddd \phi_j }{ \ddd \theta }.
\end{equation*}
The emerging behavior of $\rho (\theta)$ is 
visualized in figure \ref{fig:density_1}, 
\begin{figure}[htb]
  \centering
  \includegraphics[width=0.48\textwidth]{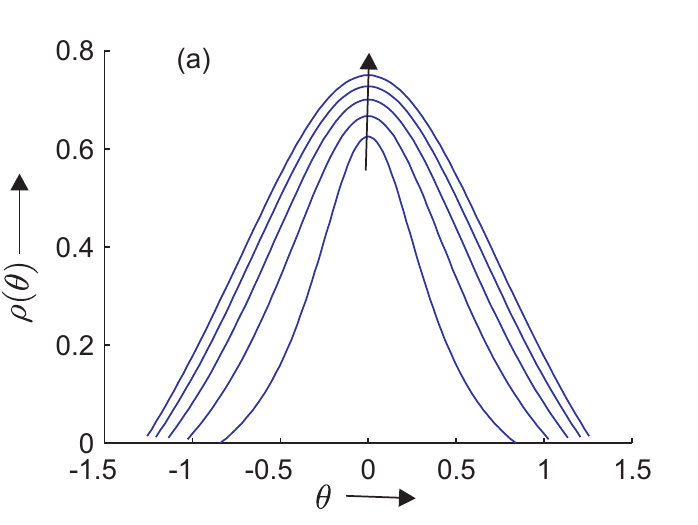}
  \includegraphics[width=0.48\textwidth]{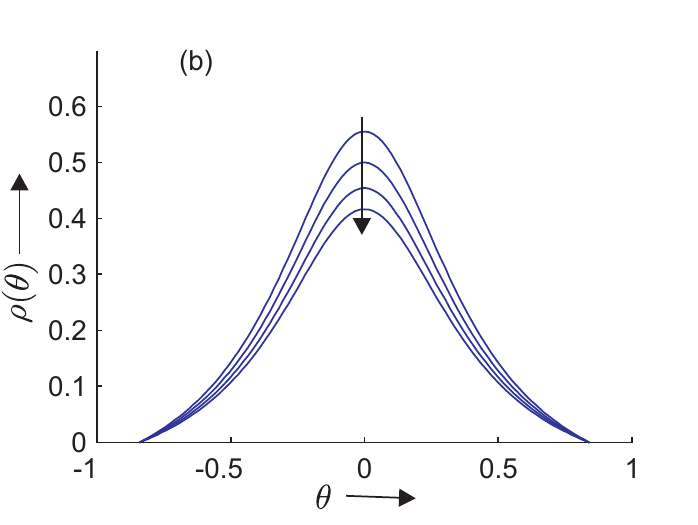} \\
  \caption{Plot of density function $\rho (\theta)$, 
           described by model 1 (minimal density pairing). 
           (a): The distance of cell bodies $\deltij > 0$ and $r_j > 0$ 
                are fixed constants, while $r_i > r_j$ successively 
                increases.
           (b): $r_i > r_j > 0$ are fixed constants, while 
                the distance of cell bodies successively increases.
          }\label{fig:density_1}
\end{figure}
where the maximal angle $\thm$, as found before in 
equation (\ref{eq:thetamax}), can be 
clearly seen in plot (b). It does not appear so prominent in 
plot (a), where it varies with $r_i$.

\paragraph{Model 2: Mean density pairing.}
Assuming that each filament from either of the neighboring cells has 
a probability to randomly form a pair at some junction on $\Gamma_{ij}$, 
the resulting pairing density can be defined as the
geometric mean of $\rho_i$ and $\rho_j$: 
\begin{equation}\label{eq:rho_theta_2}
  \rho_\text{mean} (\theta)  =  
  \sqrt{ \rho_i (\theta) \cdot \rho_j (\theta) } =
  \tilde{\rho} \cdot 
    \frac{ r_j }{ \Rij \cdot \sqrt{ \kappa_\eta (\phi_j) } } \cdot
    \frac{ \ddd \phi_j }{ \ddd \theta }.
\end{equation}
In figure \ref{fig:density_2} the emerging behavior of $\rho (\theta)$ is 
visualized, showing an even more expressed cut-off at $\theta = \thm$
\begin{figure}[htb]
  \centering
  \includegraphics[width=0.48\textwidth]{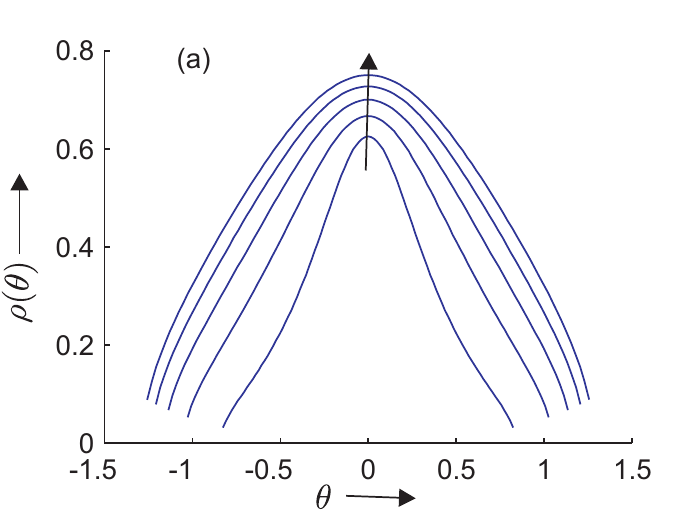}
  \includegraphics[width=0.48\textwidth]{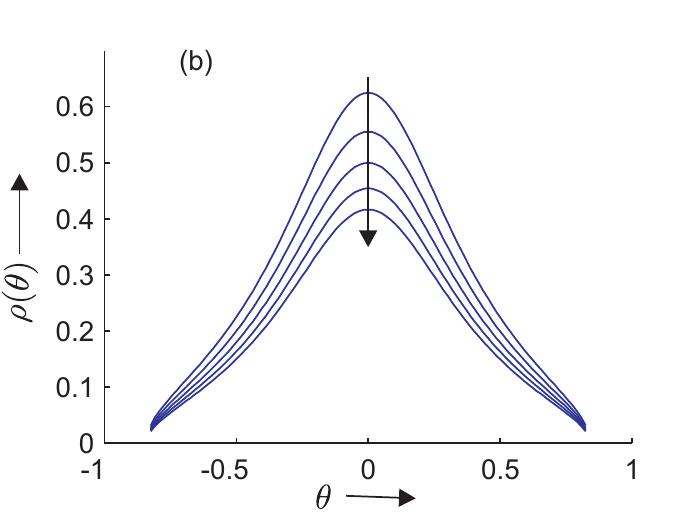}
  \caption{Plot of density function $\rho (\theta)$, 
           described by model 2 (mean density pairing).
           (a): The distance of cell bodies $\deltij > 0$ and $r_j > 0$ 
                are fixed constants, while $r_i$ successively 
                increases.
           (b): $r_i > 0$ and $r_j > 0$ are fixed constants, while 
                the distance of cell bodies $\deltij$ successively 
                increases.
          }\label{fig:density_2}
\end{figure}

From figures \ref{fig:density_1} and \ref{fig:density_2} it 
becomes apparent that $\rho (\theta)$ is even in $\theta$, 
maximal for $\theta = 0$, and strictly decreasing 
for increasing $|\theta|$ in both methods. 
Whereas model 1 captures the maximal filament pairing 
that could be realized for long term association 
at fixed adherens junctions, model 2 describes 
the pseudo-steady state of 
short term stochastic filament association, 
which will be considered further on.

\subsection{Pair interaction force}\label{sec:fint}

Consider two cells $i,j$ touching each other, so that 
if their cell body distance $\deltij = \delta_i + \delta_j$ 
was further increased, they would dissociate. 
By equation (\ref{eq:lamwid}) this limiting 
rupture distance $\deltru$ is given as
\begin{equation}\label{eq:deltru}
  \deltru = \deltif + \deltjf = 
            \bigl( \smcPm - 1 \bigr) ( r_i + r_j ).
\end{equation}
Then according to the previous assumptions, 
any paired couple of actin fibers meeting at an adherens junction 
in the contact boundary $\Gamma_{ij}$ 
near the intersection point $\vSij$ (see figure \ref{fig:ratrule}) 
develops a certain positive stress between the two cell bodies. 
According to the assumption made at the beginning of this chapter, 
this stress depends on the mean volume fraction $q_0$ 
of the contractile cytoskeletal network, 
which before touching was equal in both contacting lamellae
of width $\deltif, \deltjf$, respectively. 
If now $\deltij$ further decreases, 
then both lamellae will be compressed 
by the equal factor 
$\delta_i / \deltif = \delta_j / \deltjf = \deltij / \deltru < 1$ 
as a consequence of the Voronoi partition laws 
(\ref{eq:cVNG}) and (\ref{eq:closed_Gammaij}). 
Thus, we can suppose that 
(a) the mean volume fraction in both lamellae 
    increases to the same value $q$ satisfying the inverse relation 
    \begin{equation}\label{eq:voldelt}
      \frac{q}{q_0} = \frac{\deltru}{\deltij},
    \end{equation}
    and
(b) any paired actin fibers develop the same stress between their 
    adherens junction and the corresponding cell body, 
    with a strength $\tilde f = f (q)$ that, for simplicity, 
    depends only on the common cytoskeletal volume fraction $q$.
Since the cytoskeletal network consists not only of cross-linked 
actin-myosin filaments but also of more or less flexible 
microtubuli and intermediate filaments (as keratin, for example) 
\cite{Koestler2008,Taute2008,Sivaramakrishnan2008,Alberts}, 
the stress function $f(q)$ has to decrease to (large) negative 
values for increasing $q \rightarrow q_\text{max} = 1$. 
Here we chose the simple, thermodynamically compatible 
strictly decreasing model function 
\begin{equation*}
  f(q) = \finth \Bigl( \ln (1-q) - \ln q - \ln z_c \Bigr).
\end{equation*}
The corresponding convex generalized free energy $\mathcal{F}$ 
satisfies 
\begin{equation*}
  \mathcal{F} (1-q) = (1-q) \Bigl( f(q) - \finth \Bigr)
\end{equation*}
for $0 < q < 1$ (\cf \cite{AltAlt2008}), 
where the positive constant $z_c < 1 / q_0 - 1$ 
determines the critical volume fraction 
$q_c = 1 / ( 1 + z_c ) > q_0$ such that $f(q_c) = 0$. 
Applying transformation (\ref{eq:voldelt}) we finally obtain 
an actin fiber stress function that depends only 
on the relative cell body distance $\Delij = \deltij / \deltru < 1$, 
namely 
\begin{equation}\label{eq:fcyt}
  f(\Delij) = \finth \cdot \ln \Bigl( 
      \frac{ \Delij - \Delmi }{ \Delcr - \Delmi } 
    \Bigr),
\end{equation}
where $0 < \Delmi = q_0 < q_0 ( 1 + z_c ) = \Delcr < 1$.

The derivation of this stress model relies on 
the simplifying assumption that 
according to equation (\ref{eq:voldelt}) 
the stress of each paired filament 
extending from cell body $\mcBri$ to 
the adherens junction at $\Gamma_{ij}$ 
is completely determined 
by the adhesion strength 
(appearing as coefficient \finth) 
and the cytoskeletal state $q$ 
of the intermediate lamella near 
the horizontal cell-cell connection axis in direction $\udij$, 
see figures \ref{fig:pmax} and \ref{fig:Density}. 
Moreover, relative to this coordinate frame 
the paired filament orientations are 
$\uRi = ( - \cos \phi_i, \sin \phi_i )$ and 
$\uRj = (   \cos \phi_j, \sin \phi_j )$, 
so that the corresponding adherens junction at $\Gamma_{ij}$ 
experiences two force vectors 
$\vffi = - f (\Delij) \cdot \uRi$ and $\vffj = - f (\Delij) \cdot \uRj$ 
with opposing horizontal components.
However, their resultant vector $\vffi + \vffj$ 
generally does not vanish (except for $\phi_i = \phi_j = 0$). 
It has a negative vertical component 
$- f (\Delij) \cdot ( \sin \phi_i + \sin \phi_j )$,
which could pull the adherens junction 
towards the cell-cell connection line 
along the contact boundary $\Gamma_{ij}$.

Therefore, some counterforces 
due to substrate adhesion via \eg integrin \cite{Friedl1998,Hegerfeldt2002} 
or frictional drag have to be supposed in order to guarantee 
the assumed pseudo-stationary equilibrium condition for $\Gamma_{ij}$.
Using the simplifying decomposition 
in horizontal and vertical components, 
we arrive at the following model expression 
for the force $\vffij$ applied by a single filament pair 
onto the cell body center $\vxj$:
\begin{equation}\label{eq:vffij}
  \begin{split}
  \vffij &= \frac{1}{2} \bigl( \vffi - \vffj \bigr)^\text{(hor)} -
           \frac{\alpha}{2} \bigl( \vffi + \vffj \bigr)^\text{(ver)} \\
         &= \frac{ f (\Delij) }{2} \Bigl(
                    ( \cos \phi_i + \cos \phi_j ) \udij +
             \alpha ( \sin \phi_i + \sin \phi_j ) \udij^\perp
           \Bigr),
  \end{split}
\end{equation}
where $\alpha \geq 0$ is an additional adhesion or friction parameter. 
By relying on the pairing filament density $\rho (\theta)$  
in section \ref{sec:filapair}, 
we obtain an integral expression 
for the total pair interaction force applied by cell $i$ onto 
cell $j$:
\begin{equation}\label{eq:interaction_force}
  \vFiji = R_{ij} \int_{ \Gamma_{ij} } \ddd \theta \;
          \rho (\theta) \cdot \vffij (\theta)
\end{equation}
where the trigonometric relations between $\phi_i$, $\phi_j$ and 
the parameterization angle $\theta$ have to be extracted from equations 
(\ref{eq:map_rho}) -- (\ref{eq:fdReta}). Conversely, the force of 
cell $j$ onto $i$ is determined by the relations 
\begin{equation}\label{eq:fintsym}
  \Fjihor = - \Fijhor \qquad \Fjiver = \Fijver.
\end{equation}
The emerging cell pair interaction force as described by 
equation (\ref{eq:interaction_force}) 
is shown in figures \ref{fig:fiji3d} and \ref{fig:fijicut}.
\begin{figure}[htbp]
  \centering
  \includegraphics[width=0.60\textwidth]{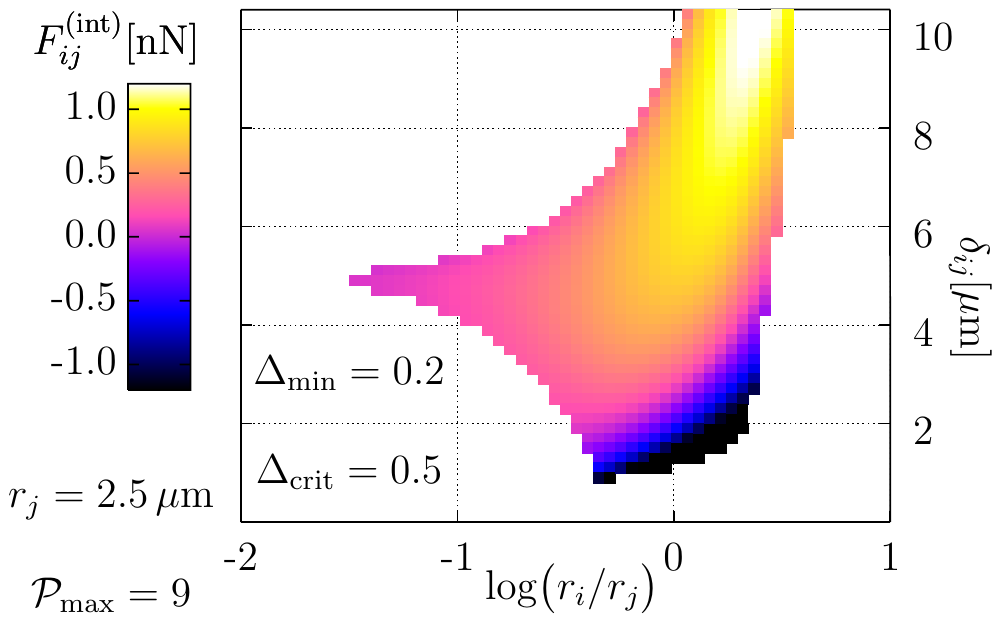}
    \caption{Modulus $F_{ij}^\text{(int)}$ 
             of pair interaction force $\vFiji$, see 
             equation (\ref{eq:interaction_force}). In the
             empty regions of the plot, $\vFiji$ is not defined. There, 
             $\mcBRif \cap \mcBRjf$ extends beyond $\Gamma_{ij}$ as 
             bounded by $\thm$, or the cells are not in contact at all.
            }\label{fig:fiji3d}
\end{figure}
\begin{figure}[htbp]
  \centering
  \includegraphics[width=0.48\textwidth]{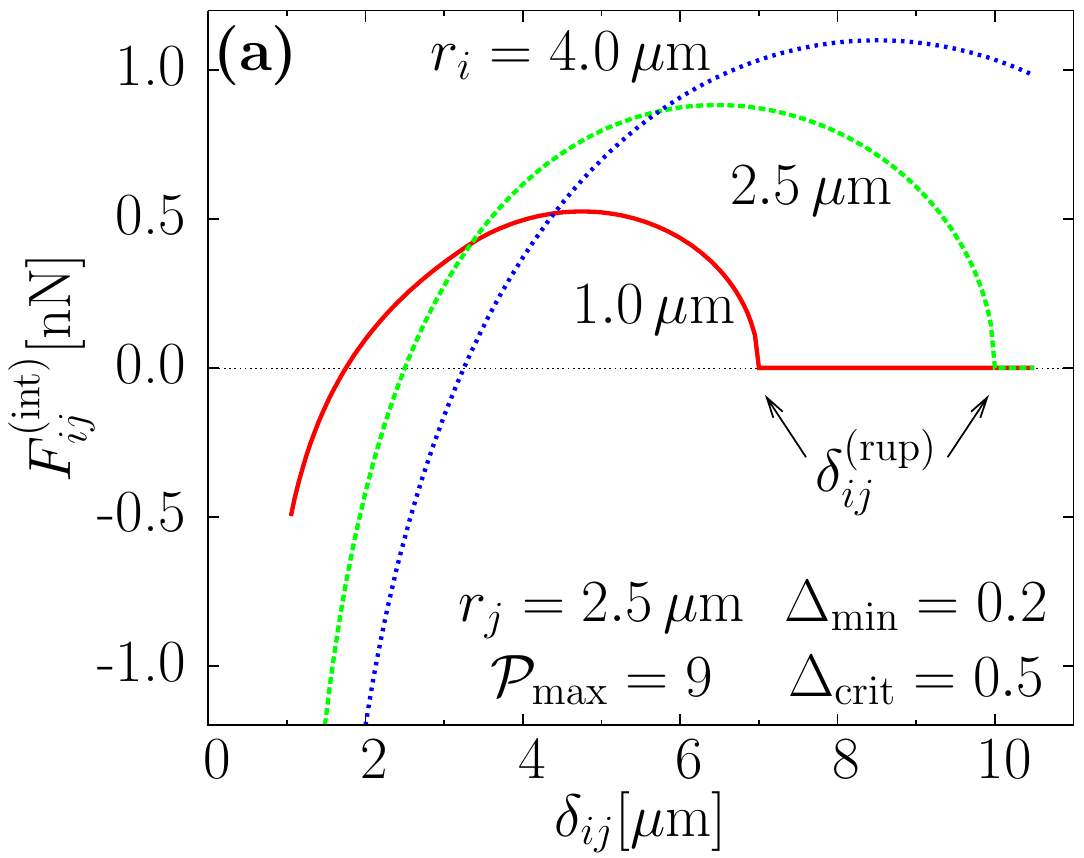} \quad
  \includegraphics[width=0.48\textwidth]{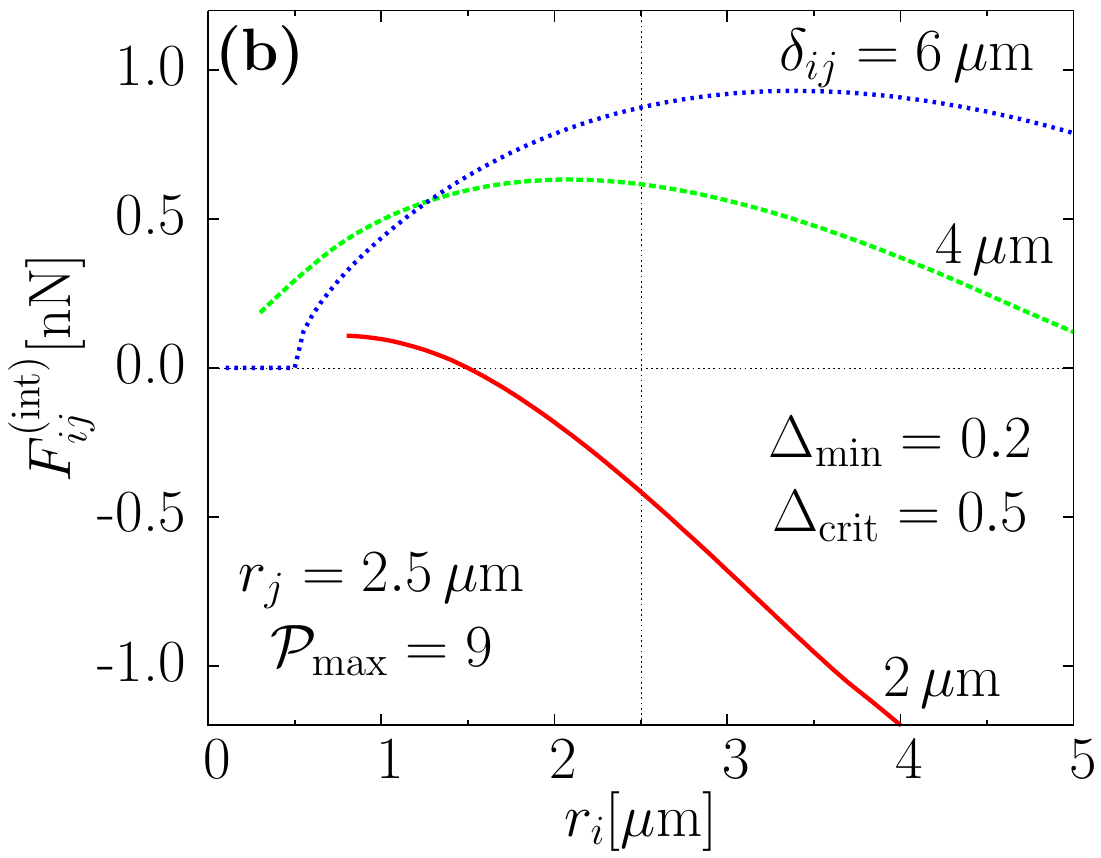}
  \caption{Modulus of pair interaction force $\vFiji$ 
           depending on (a) $\deltij$ and (b) $r_i$. 
           The global view depending on both $\deltij$ 
           and the logarithmic ratio $r_i / r_j$ was presented in figure 
           \ref{fig:fiji3d}.
          }\label{fig:fijicut}
\end{figure}
A natural and maximal cut-off distance for the force is given by  
the finiteness of the Voronoi tessellation, whereby neighboring is only 
possible for sufficiently small cell center distances  
$ \deltij < \deltru + r_i + r_j$, \ie 
$\overline{\mcBRif} \cap \overline{\mcBRjf} \neq \emptyset$.
Once two previously isolated cells come close enough 
for contact, there is a strong tendency to attach, 
which facilitates multicellular tissue formation. 
The interaction force is attractive until the cell distance 
$\deltij$ reaches $\deltcr = \Delcr \cdot \deltru$, 
where $\vFiji$ vanishes. 
Finally, if $\deltij$ drops below $\deltcr$, 
$\vFiji$ becomes repulsive and therefore hinders tissue collapse 
at distances approaching $\deltmi = \Delmi \cdot \deltru$.
Note that with $\Delij > \Delmi$ 
the lower bound from inequality (\ref{eq:pmax_max2}) 
on the homogeneity of cell radii due to fixed $\sqrt{\mcPm}$ 
can be relaxed to 
\begin{equation}\label{eq:pmax_max}
  1 < \frac{ \mcPm  }{ \bigl( \Delmi ( \sqrt{ \mcPm } - 1 ) + 1 \bigr)^2 }
    \leq \qnb. 
\end{equation}
Correspondingly, 
$\mcPm$ can be increased for given cell homogeneity $\qqq$ or $\qnb$.
For example, the constraint (\ref{eq:pmax_max}) 
yields $\qnb = 6.25$ for $\sqrt{\mcPm} = 3$, or 
$r_\text{min} \geq 0.73 \cdot r_\text{max}$ for each cell pair. 
In fact, the actual distances $\Delij$ in a tissue will be 
larger than $\Delmi$, effectively relaxing (\ref{eq:pmax_max}) even 
further.

\subsection{Locomotion force at the free boundary}\label{sec:floc}

In addition to the dynamics induced by pair interaction forces, cells 
at the tissue margin may migrate into open space. 
The locomotion force causing such a migration 
is due to lamellipodial protrusion and retraction, 
which is unhindered only 
at the free cell boundary $\Gamma_{i0}$.
In a similar manner as before, we assume that this locomotion or 
free boundary force onto the cell body $\mcBri$ is determined by 
connecting radial filament bundles as indicated in figure \ref{fig:Density}. 
The filament density of cell $i$ along its free 
boundary $\Gamma_{i0}$ is given by
\begin{equation}\label{eq:free_boundary_density}
  \rhofi = \frac{ \rhon r_i }{ \Rif } = \frac{\rhon}{\sqrt{\mcPm}},
\end{equation}
and thus independent of $r_i$.
In this way, the locomotive force of a cell $i$ reads as
\begin{equation}\label{eq:free_boundary_force}
  \vFifb = \floc \int_{\Gamma_{i0}} \ddd s_i \,
               \rhofi \uRif (\phi_i), 
\end{equation}
with arc length $s_i = \Rif \phi_i$ and 
$\uRif (\phi_i) = (\cos \phi_i, \sin \phi_i)$. 
Moreover, in order to heuristically account for ubiquitous perturbations 
due to lamellipodial fluctuations or possible signals, 
we implement stochastic force increments at the tissue margin 
\begin{equation}\label{eq:stochastic_force}
  \ddd \vFist = b_0 \int_{\Gamma_{i0}} \ddd \mathbf{B}_{t, s_i}. 
\end{equation}
Here we assume a uniform and anisotropic vector 
noise $\mathbf{B}_{t, s_i}$ defining a 
spatio-temporal Brownian sheet 
in arc length and time coordinates 
with independent Gaussian increments 
satisfying $\Var \bigl( \ddd \mathbf{B}_{t s_i} \bigr) = 
\ddd s_i  \cdot \ddd t$.
For each time $t$, stochastic integration 
results in a simple weighted Gaussian noise term 
with random increments $\ddd \mathbf{W}_t$ 
\begin{equation}\label{eq:Wdiscr}
  \ddd \vFist = b_0 \sqrt{|\Gamma_{i0} |} \ddd \mathbf{W}_t
  \; \widehat = \;
    b_0 \sqrt{ \frac{ |\Gamma_{i0}| \ddd t}{2} } \pmb \xi_t,
\end{equation}
where $\pmb \xi_t$ is a vector of Gaussian random numbers, 
which is chosen independently for every time step 
in a corresponding numerical realization of the stochastic process.

\subsection{Drag forces}

Apart from interaction and free boundary forces, the cell 
is subject to drag forces $\vFifr$ slowing down its movement.
Such drag forces are generally functions of 
the cell body velocity $\dot{\vx}_i = \vvi$. 
Here we assume the simplest dependency of a linear 
force-velocity relation
\begin{equation}\label{eq:drag_force}
  \vFifr = - \gamma_i \vvi,
\end{equation}
with drag coefficient $\gamma_i = \gamma (r_i)$. Arising from friction 
with the substratum, $\gamma_i$ could depend on the area of the 
cell body, \eg $\gamma (r_i) \propto r_i^2$, however, 
for simplicity, here we take $\gamma_i = \tilde{\gamma}$ 
independent of cell body sizes.

\subsection{Dynamics of cell movement}\label{sec:dynamics}

The previously described, active and anisotropic forces 
$\vFiji, \vFifb$ arising from the actin filament network 
act onto the cell center $\vxi$
causing translocation of the cell.
However, friction, see equation (\ref{eq:drag_force}), 
is considered to be dominating 
and inertia terms are neglected \cite{Galle2005,Schaller2005,Honda2004},  
so that the emerging deterministic overdamped 
Newtonian equations of motion read as 
\begin{equation}\label{eq:newton}
  \vvi = 
  \frac{1}{\gamma_i} \biggl( 
    \vFifb + \sum_{j \, \text{neighbor}} \vFjii
  \biggl) =: \frac{ \vFi }{ \gamma_i }.
\end{equation}
Moreover, any change of the translocation direction as well as adjustment 
of speed to the pseudo-steady state as given by the previous equation 
(\ref{eq:newton}) requires some (mean) time $T_i$ for  
restructuring and reinforcing the anisotropic actin network.
The simplest way to model this adjustment process 
is by a linear stochastic filter of first order for the velocity 
\cite{Arnold1974}. 
Together with equation (\ref{eq:stochastic_force}) this results in 
the stochastic differential equation (SDE) system 
\begin{equation}\label{eq:dynamics}
  \ddd \mathbf v_i = 
    \frac{1}{T_i} \biggl( \frac{\vFi}{\gamma_i} - \mathbf v_i \biggl) \ddd t +
    b_i \sqrt{| \Gamma_{i0} |} \ddd \mathbf W_t,
  \qquad \qquad
  \ddd \vxi = \mathbf v_i \ddd t, 
\end{equation}
with $b_i = b_0 / \gamma_i$.
Similarly as the friction $\gamma_i$, also the mean adjustment 
time $T_i$ could have some dependence on $r_i$, however, here 
we restrict ourselves to the case of cells with homogeneous 
activity time scale $\forall i: T_i = T$.

For each time $t$, the forces 
(\ref{eq:interaction_force},\ref{eq:free_boundary_force},\ref{eq:Wdiscr},\ref{eq:drag_force})
can be computed explicitly from the Voronoi tessellation of
the generating cell bodies $\{ \mcBri: i = 1 \dots N \}$ 
using a spatial discretization of $\{ \Gamma_{ij} \}$ 
in the parameterizing angle $\theta$.
Next, the velocities $\{\vvi\}$ and positions $\{\vxi\}$ 
of the cell centers are updated according to 
both equations (\ref{eq:dynamics}) in an 
explicit Euler-Maruyama step \cite{Kloeden1999}.
Finally, the Voronoi tessellation is computed anew from the 
updated generators $\{ \mcBri \}$.

Higher order stochastic integration schemes were not applied, 
since such procedures necessitate the distribution of both 
forces and perturbations onto the powers of a Taylor expansion. 
In the general case, the involved derivatives 
of cell-cell contacts $\{ \Gamma_{ij} \}$ and 
cell margins $\{ \Gamma_{i0} \}$
cannot be computed easily \emph{a priori}. 
In particular, the change of a contact $\Gamma_{ij}$ may 
depend on the behavior of several distinct nearby cells $k \neq i,j$. 
Altogether, here we use the versatile basic method 
for integrating the equations of motion, 
because it is applicable regardless of the structure 
of the underlying SDE system.

\section{Cell pair contacts}\label{sec:pair_contacts}

From the force plots in figure \ref{fig:fijicut} it is clear that 
the interaction force $F_{ij}^\text{(int)}$ between a cell pair 
exhibits a sharp onset 
when two formerly dissociated cells come into contact. 
For any such cell pair in contact and each cell-cell body 
distance $\delta = \deltij < \deltru$, see equation (\ref{eq:deltru}), 
there is a unique pair of maximal contact angles $\phif, \phjf$.
Therefore, according to equations
(\ref{eq:Rif}) and (\ref{eq:phieta}) with $\eta = r_j / r_i$ 
and figure \ref{fig:pmax} the relation 
\begin{equation*}
  \Rjf \cdot \sin \bigl( \phjf (\delta) \bigr) = 
  \Rif \cdot \sin \bigl( \phif (\delta) \bigr).
\end{equation*}
holds and 
the free boundaries of both cells $k=i,j$ are characterized by
\begin{equation*}
  \Gamma_{k0} = \Bigl\{ 
    \Rkf \cdot \bigl( 
      \cos \phi_k, \sin \phi_k
    \bigr): \; |\phi_k| \geq \phi_{k0} (\delta)
  \Bigr\}.
\end{equation*}
Thus, by solving the integral 
in equation (\ref{eq:free_boundary_force}) and 
regarding relation (\ref{eq:free_boundary_density}) 
we obtain for the deterministic part of the locomotive forces
\begin{equation*}
  \begin{split}
    \vFifb &= 2 \floc \rhon r_i \cdot  
               \sin \bigl( \phi_{i0} (\delta) \bigr) \udij \\
            &= 2 \floc \rhon r_j \cdot
               \sin \bigl( \phi_{j0} (\delta) \bigr) \udij = - \vFjfb, 
  \end{split}
\end{equation*}
again using the decomposition in horizontal and vertical components.
This means, that under our model conditions 
the mean locomotive forces of the two cells 
are exactly opposite, independent of their size. 
Since the interaction forces have the same property 
(see equation (\ref{eq:fintsym}) with $\Fijver = 0$), we conclude that 
the sum of the deterministic driving forces 
onto the cell pair vanishes, $\vFi + \vFj = 0$. 
From equation (\ref{eq:rho_theta_2})  
the interaction force onto cell $j$ is computed as 
\begin{equation}\label{eq:touch_vFiji}
  \vFiji = 2 \rhon r_j f \biggl( \frac{\delta}{\deltru} \biggr) \cdot
           C \bigl( \phi_{j0} (\delta) \bigr) \udij, 
\end{equation}
with $C (\phi) = \int_0^\phi \ddd \varphi \, 
\bigl( \cos \varphi + \sqrt{ 1 - \eta \sin^2 \varphi } \bigr) / 
\kappa_\eta (\varphi)^{1/2}$, 
and $\kappa_\eta$ as defined in equation (\ref{eq:dphidth}). 
\begin{figure}[htb]
  \centering
  \includegraphics[height=0.33\textwidth]{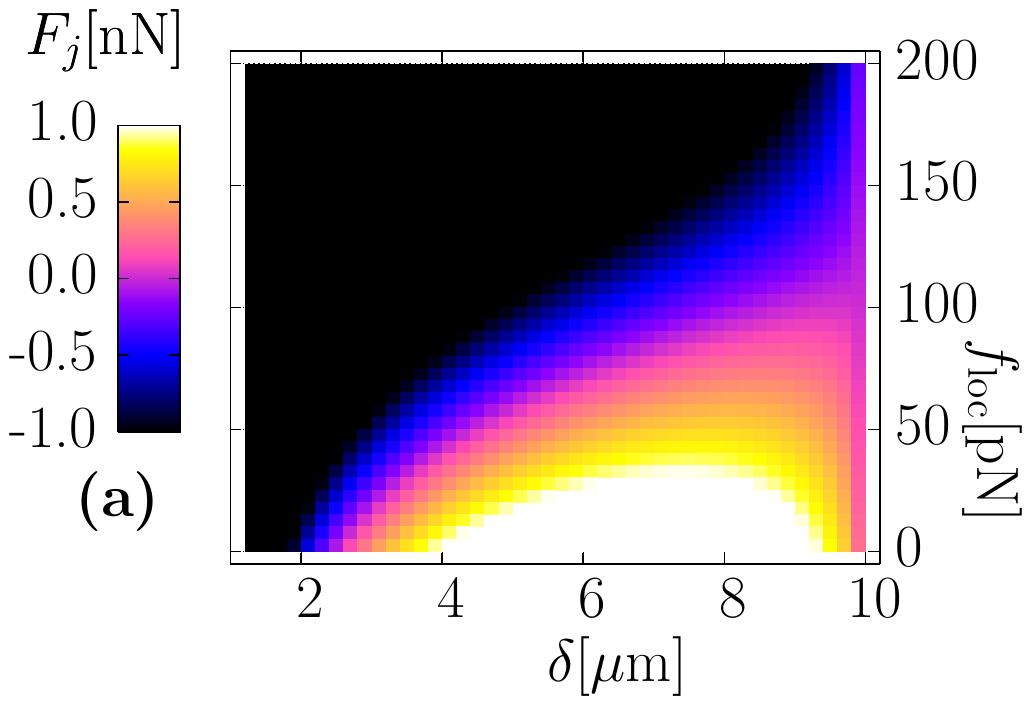} \quad
  \includegraphics[height=0.33\textwidth]{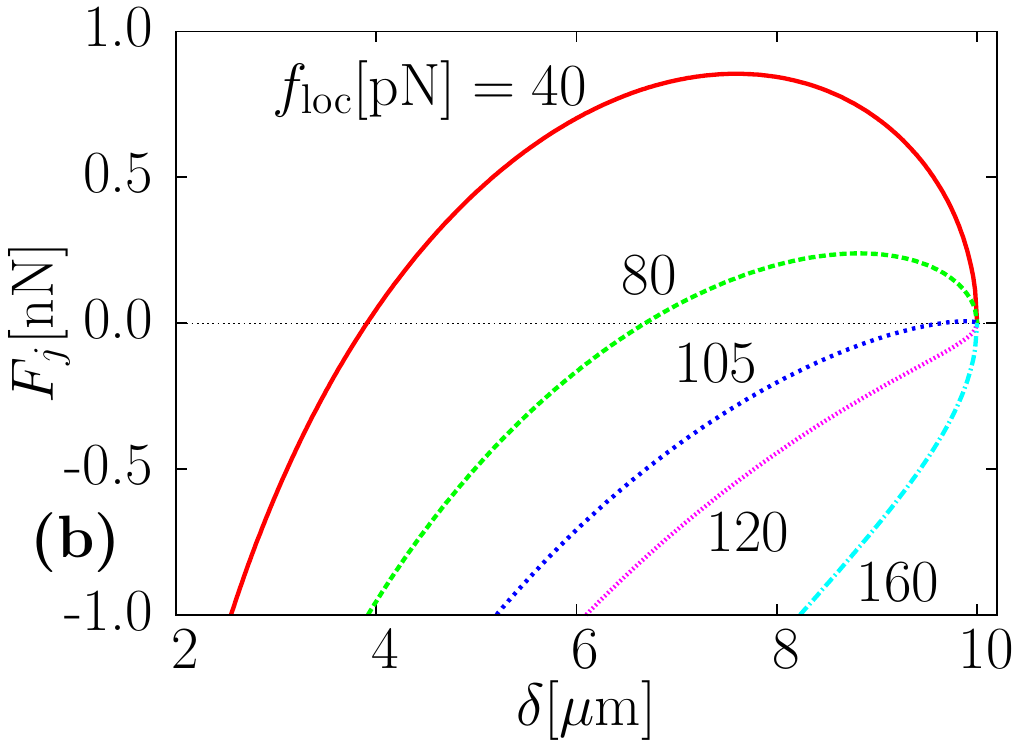}
  \caption{Force balancing in an isolated cell pair. 
           The deterministic force 
           $F_j = F_{ij}^\text{(int)} + F_j^\text{(loc)}$ 
           acts on the cell center $\vxj$.
           Parameters are $f_\text{int} = 60 \, \mathrm{pN}$, $\alpha = 1$, 
           $r_i = 3.0$, $r_j = 2.0$, $\rhon = 6.6 / \mathrm{\mu m}$, 
           $\Delmi = 0.2$, $\Delcr = 0.5$. 
          }\label{fig:deltfloc}
\end{figure}
In figure \ref{fig:deltfloc} the resulting scalar horizontal force 
$F_j = F_{ij}^\text{(int)} + F_j^\text{(loc)}$ 
is plotted as a function of cell body distance $\delta$ 
and the locomotion force parameter $\floc$, revealing the 
emergence of a stable deterministic contact equilibrium 
$F_{ j, \text{det} } = 0$ for lower values of $\floc < 105 \, \mathrm{pN}$. 
On the other hand, for larger locomotion parameters, 
no such equilibrium exists and the cell distance always increases 
until the pair separates at $\delta = \deltru$ 
($= 10 \mu \mathrm{m}$ in figure \ref{fig:deltfloc}).

In order to study the full stochastic contact and segregation dynamics 
described by the SDE system (\ref{eq:dynamics}) with noise 
amplitudes $\tilde b_k = b_0 \sqrt{ G_k ( \phi_{k0} (\delta) ) } / \gamma_k$, 
$G_k (\phi) = |\Gamma_{k0}| = 2 \sqrt{\mcPm} r_k (\pi-\phi)$, 
we set all vertical noise components to zero, for simplicity. 
Then due to $F_i + F_j = 0$ the 
dynamics is determined by differential equations for 
the cell overlap $z$ and the difference 
$u$ in horizontal cell velocities:
\begin{equation}\label{eq:overlap}
  z = \deltru - \delta > 0, \qquad 
  u = v_j^\text{(hor)} - v_i^\text{(hor)}.
\end{equation}
\begin{proposition}\label{prop:overlap}
  For cell pair dynamics restricted to the horizontal connection line 
  the stochastic ODE system in equation (\ref{eq:dynamics}) transforms 
  into a nonlinear Ornstein-Uhlenbeck system  
  for the overlap $z$ and its temporal change $u$, namely
  \begin{eqnarray}
    \ddd z &=& u \, \ddd t \label{eq:zornstein} \\
    \ddd u &=& \Bigl( F(z) - u \Bigl) \frac{ \ddd t }{ T } +
               b(z) \cdot \ddd W_t \label{eq:wornstein}
  \end{eqnarray}
  Defining the mean drag coefficient 
  $1 / \gamma_{ij} := 1 / \gamma_i + 1 / \gamma_j$ and
  $C(\phi)$ as in equation (\ref{eq:touch_vFiji}), we have 
  \begin{gather}
    \label{eq:forceFz}
    F (z)  = \frac{ 2 \rhon r_j }{ \gamma_{ij} } \Bigl( 
                 f (z) C (\phi) - \floc \cdot \sin \phi 
               \Bigr), \quad
    f(z) = \finth \log \Bigl( \frac{\zzr-z}{\zzr-z_c} \Bigr), \\
    \label{eq:forcefbz}
    b^2 (z) = b_0^2 
              \biggl( \frac{ G_j (\phi) }{ \gamma_j^2 } + 
                      \frac{ G_i (\phi) }{ \gamma_i^2 } \biggr),
  \end{gather}
  with $\zzr = \deltru \cdot ( 1 - \Delmi )$ and
  $z_c = \deltru \cdot ( 1 - \Delcr )$. 
  The relation between $z = \deltru - \delta$ and  
  $\phi = \phjf (\delta)$ can be written as
  \begin{equation}\label{eq:zofphi}
    z \equiv \Rjf \biggl( 1 - \sqrt{ 1 - \sin^2 \phi } \biggr) +
             \Rif \biggl( 1 - \sqrt{ 1 - \eta^2 \sin^2 \phi } \biggr).
  \end{equation}
\end{proposition}
Note that the overlap $z$ is a monotone function of $\sin \phi$, 
which is proportional to the vertical extension 
of the overlap region as spanned by the contact arc $\Gamma_{ij}$
(shaded in figure \ref{fig:pmax}).

\subsection{Asymptotic stochastic differential equations}

Disruption of a connected pair occurs as $z \to 0$, 
so that an expansion at zero of all terms 
in Proposition \ref{prop:overlap} is justified. 
First, from equation (\ref{eq:zofphi}) we derive 
the asymptotic relation 
\begin{equation*}
  z = 
  \Rjf \frac{ 1 + \eta }{2} \sin^2 \phi \Bigl(
    1 + \mathcal{O} \bigl( \sin^2 \phi \bigr)
  \Bigr)
\end{equation*}
for $z>0$ so that
\begin{equation*}
  \phi \sim \sin \phi = 
  \sqrt{ \frac{ 2 }{ (1+\eta) \Rjf } } \cdot \sqrt{z} \cdot \Bigl(
    1 + \mathcal{O} (z)
  \Bigr).
\end{equation*}
Thus, the locomotion term of the force $F(z)$ in equation (\ref{eq:forceFz}) 
has a singularity at zero like $\sqrt{z}$. The same holds 
for the interaction term; as a surprise, the corresponding integral 
can be expanded in $\phi$ independent of the ratio $\eta = r_j / r_i$:
\begin{equation*}
  C (\varphi) = \int_0^\phi \ddd \varphi \, \Bigl( 
    2 - \varphi^2 + \mathcal O \bigl( \varphi^4 \bigr)
  \Bigr) =
  2 \cdot \sin \phi \Bigl( 1 + \mathcal O \bigl( \phi^4 \bigr) \Bigr).
\end{equation*}
The conclusion is that the horizontal pair force can be approximated as
\begin{equation*}
  F(z) = 2 \frac{ \rhon r_j }{ \gamma_{ij} } \sin \phi 
        \biggl(
          2 f(z) \Bigl( 1 + \mathcal O \bigl( z^4 \bigr) \Bigr) - \floc 
        \biggr),
\end{equation*}
where only the prefactor depends on the cell body sizes. 
The deterministic equilibrium overlap $z_* > 0$, as the zero of $F$,  
is approximately determined by the force equality
\begin{equation*}
  f \bigl( z_* \bigr) = 
  \frac{ \floc }{ 2 } \cdot 
    \Bigl( 1 + \mathcal O \bigl( z_*^2 \bigr) \Bigr).
\end{equation*}
Thus, by using the definition of $f$ in equation (\ref{eq:forcefbz}), we 
define the \emph{contact parameter} 
\begin{equation}\label{eq:contact_parameter}
  \lambda := 
  \zzr - (\zzr-z_c) \exp \biggl( \frac{ \floc }{ 2 \finth } \biggr).
\end{equation}
Let us consider the deterministic 
overdamped dynamics $\dot z = F(z)$ obtained from the 
stochastic ODE system (\ref{eq:zornstein}, \ref{eq:wornstein}) 
in the limit $T \to 0$. With $\gamma_i = \gamma_j = \tilde \gamma$ 
we prove the 
\begin{proposition}[Asymptotic contact and disrupture dynamics] 
  Given a pair of cells with body radii $r_i$ and $r_j$, 
  with small contact parameter $|\lambda|$. 
  Then there exists a unique stable equilibrium $z_* > 0$ 
  if and only if $\lambda > 0$, namely  
  $ z_* = \lambda \bigl( 1 + \mathcal O ( |\lambda|^2 ) \bigr) > 0$.
  Moreover, in the limit 
  $T \to 0$, $\beta_0 = b_0 / \tilde \gamma T = \mathrm{const.}$,
  and the corresponding stochastic differential equation (SDE)
  can be approximately written as
  \begin{equation}\label{eq:rupdyn}
    \ddd z = 
    \mu \sqrt{ z_+ } \ln \biggl(
      \frac{ \zzr - z }{ \zzr - \lambda }
    \biggr) \ddd t +
    \beta (z) \ddd W_t.
  \end{equation}
  Here $z_+ = \max (0,z)$ and 
  \begin{equation*}
    \mu = \rhon \frac{ G_{ij} }{ 2 \tilde \gamma } \cdot \finth, \qquad
    \beta (z) = 
    \beta_0 \biggl( 2 \sqrt{\mcPm} \Bigl( 
        ( r_i + r_j ) \pi - G_{ij} \cdot \sqrt z 
    \Bigr) \biggr)^{1/2},
  \end{equation*}
  with 
  $G_{ij} = \bigl( 8 \sqrt{\mcPm} \cdot r_i r_j / ( r_i + r_j ) \bigr)^{1/2}$.
\end{proposition}
Note that the $\log$ term in 
equation (\ref{eq:rupdyn}) can be approximated 
by $\bigl( \lambda - z + \mathcal O (z^2) \bigr) / \zzr$, 
and $z^* = \lambda$ is a measure of the mean cell-cell overlap.

As soon as the contact between cells is lost, $z<0$, their 
cell center distance $d_{ij} = \sqrt{\mcPm} (r_i+r_j) - z$ 
would perform a pure Brownian motion for $T=0$ or, 
for $T>0$, a persistent random walk with 
the standard recurrence probability to hit 
the touching state $z=0$ again. 
However, for situations of tissue cells moving 
on two-dimensional substrates, the production of adhesive fibers 
(as fibronectin) or remnants of plasma membrane plus adhesion molecules 
in the wake of a migrating cell would induce a positive bias of the 
locomotion force towards the other cell \cite{Kirfel2003}, 
which could be assumed as 
proportional to the cell boundary distance $-z$, 
at least for small distances. 
Therefore, and for the aim of exploring the resulting 
stationary process, instead of equations (\ref{eq:rupdyn}) 
we consider the extended SDE model
\begin{equation}\label{eq:sde_extended}
  \ddd z = \tilde F (z) \ddd t + \beta (z) \ddd W_t
\end{equation}
with
\begin{equation*}
  \tilde F (z) = \left\{
  \begin{aligned}
    & \mu \sqrt z \ln \biggl( \frac{\zzr - z}{\zzr - \lambda} \biggr)
      \quad &\text{for} \; z \geq 0 \\
    & - \nu z \quad &\text{for} \; z \leq 0
  \end{aligned}
  \right.
\end{equation*}
where we introduce an additional bias parameter $\nu > 0$ 
describing an indirect attraction between separated cells.

\subsection{Analysis of the stationary contact problem}

By solving the stationary Kolmogorov forward equation, we compute the 
approximate \emph{stationary probability distribution} 
for the overlap $z$ as
\begin{equation}\label{eq:distri}
  \begin{split}
    p(z) =& 
    p_0 \exp \biggl( 
      \int_0^z \ddd z \; 
        \frac{ \tilde F (z) }{ \beta^2 (z) }
    \biggr) \\
    =& \left\{
      \begin{aligned}
        &\frac{\mu}{ \psi \zzr } z^\frac{3}{2} \Bigl(
          \frac{2}{3} \lambda + \frac{\lambda}{2 \chi} z^\frac{1}{2} -
          \frac{2}{5} z 
        \Bigr) \quad & \text{for} \; z \geq 0 \\
        &- \frac{\nu}{2 \psi} |z|^2 \quad & \text{for} \; z \leq 0
      \end{aligned}
    \right.
  \end{split}
\end{equation}
with a unique normalization factor $p_0$ and additional lumped parameters 
that arise from expanding the singular noise variance 
$\beta^2 (z) = 2 \psi \bigl( 
1 - \sqrt{ z(t) } / \chi + \mathcal O (z^2(t) \bigr)$:
\begin{align*}
  \psi =& \sqrt{\mcPm} \beta_0^2 \pi (r_i+r_j)^2 \\
  \chi^2 =& 8 \pi^2  \sqrt{\mcPm} (r_i+r_j)^2 \Bigl( 
    \frac{1}{r_i} + \frac{1}{r_j} 
  \Bigr).
\end{align*}
In figure \ref{fig:distriP}(b) and (e), the probability distribution 
\begin{figure}[phtb]
  \centering
  \includegraphics[width=1.0\textwidth]{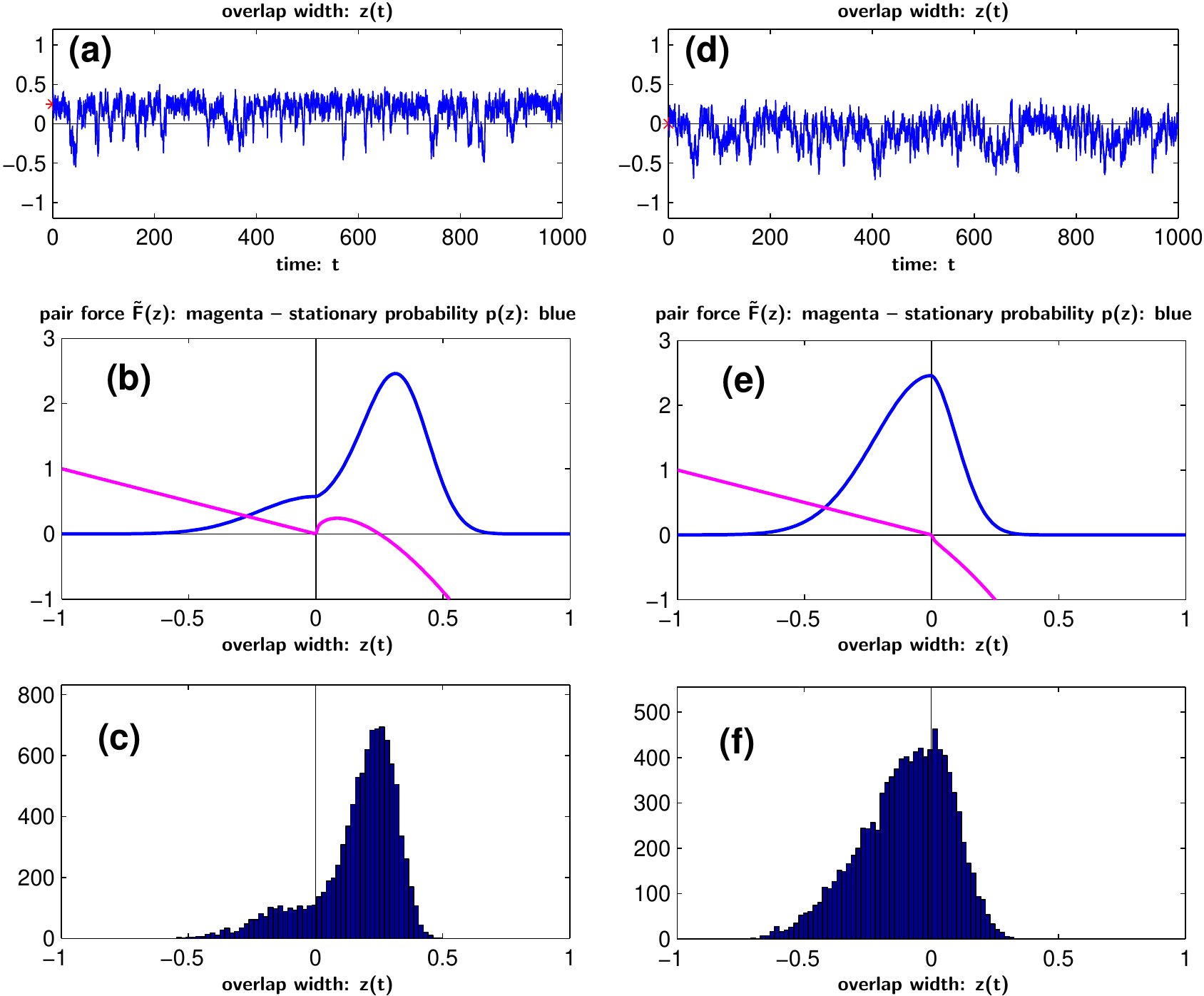}
  \caption{Analytical distribution (bold/blue in b,e), the corresponding 
           simulation histograms (c,f), and the time series 
           (a,d) of the overlap width $z$.
           This quantity emerges from the 
           dynamics of equation (\ref{eq:sde_extended}) for 
           contact parameter $\lambda > 0$ (a-c) and 
           $\lambda < 0$ (d-f), respectively, 
           with force $\tilde F (z)$ (grey/magenta in b,e). 
           Time $t$ is given in units of $\text{min}$, 
           overlap $z$ in units of $\mu \mathrm m$, 
           and force $F$ in $\mathrm{nN}$. 
           Parameters used 
           are $\beta_0 = 0.1$, $\psi = \beta_0^2$, 
           $\zzr = 0.1$ and $z_c = 0.5$.
  }\label{fig:distriP}
\end{figure}
$p(z)$ according to equation (\ref{eq:distri}) is plotted for 
two special parameter sets together with the force 
function $\tilde F (z)$, whereas in figure 
\ref{fig:distriP}(c) and (f) histograms for the corresponding 
numerical realizations of the stochastic differential equation 
(\ref{eq:sde_extended}) are shown. 
For the contact parameter $\lambda = 0.25 > 0$, 
see definition \ref{eq:contact_parameter}, there 
appears a skew-shaped modified Gauss distribution around the 
unique center $z_* = \lambda$ with a certain probability 
$P_\text{off} = \int_{-\infty}^0 \ddd z \, p(z) \approx 0.15$ 
for the cells to be separated. In contrast, for $\lambda = -0.15 < 0$, 
the standard Gauss distribution with variance $\psi / \nu$ for 
positive separation distances $-z$ is cut off by a decreasing 
distribution of the positive overlap $z$ with mean value 
$Z_\text{cont} = \int_0^\infty \ddd z \, z p(z) \approx 0.08 \ll z_*$ 
and a certain contact probability 
$P_\text{cont} = 1 - P_\text{off} \approx 0.32$.

Figure \ref{fig:distriPoff} shows the plot of the separation probability 
$P_\text{off}$ over
\begin{figure}[phtb]
  \centering
  \includegraphics[width=0.8\textwidth]{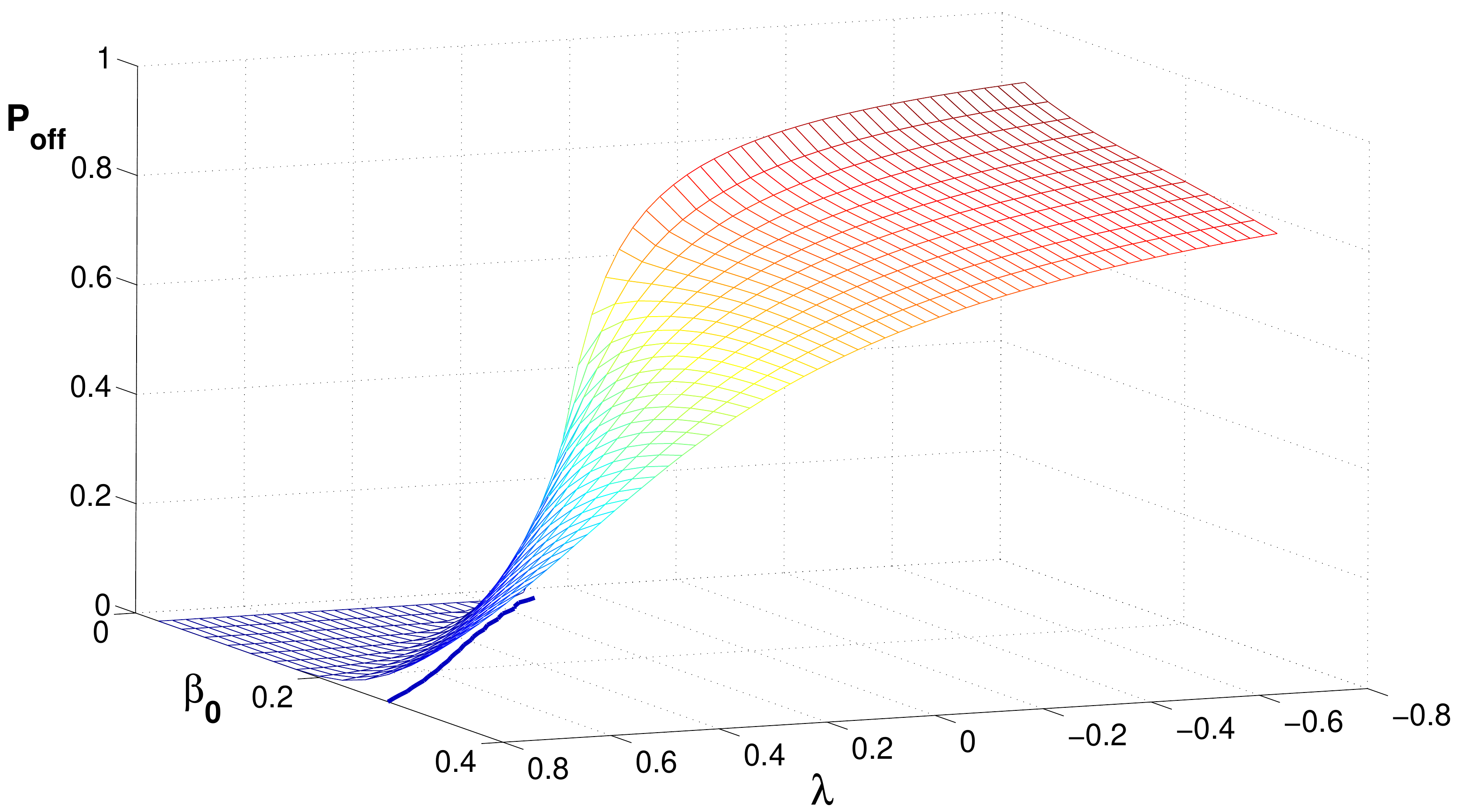}
  \caption{Separation probability 
           $P_\text{off} = \int_{-\infty}^0 \ddd z \, p(z)$ 
           according to (\ref{eq:distri})
           plotted over the parameter plane of 
           contact ($\lambda$) and noise ($\beta_0$). 
           In this plane the contour curve $P_\text{off} = 0.05$ 
           is drawn in blue.
           The other parameters are the same as in figure 
           \ref{fig:distriP}.
           More details are given in the text.
  }\label{fig:distriPoff}
\end{figure}
the contact and noise parameters $\lambda$ and $\beta_0$ 
together with the contour curve for the critical value 
$\lambda_{0.05} = \lambda (\beta_0)$, so that the separation 
probability $P_\text{off}$ is less than $5 \%$ for 
contact parameters $\lambda > \lambda_{0.05}$, \ie for 
\begin{equation*}
  \frac{ 2 \finth }{\floc} > 
  \log \biggl(
    \frac{ \Delcr - \Delmi }{ 1 - \Delmi - \lambda_{0.05} / \deltru }
  \biggr).
\end{equation*}
For the two cases of contact parameters $\lambda > 0$ and $\lambda < 0$, 
the corresponding time series of the overlap distance $z$ 
as plotted in figure \ref{fig:distriP}(a) and (d), respectively, 
reveal a clearly distinct temporal separation and contact behavior.

In the first case, with lower locomotion force parameter $\floc$ 
relative to the interaction parameter $\finth$, the contact state 
persists for longer time intervals. Thereby, the overlap fluctuates 
around the equilibrium value $z_* = \lambda$, with mean 
contact duration 
$\tau_\text{cont} \approx \exp( 2.8 ) \, \text{min} \approx 16 \, \text{min}$.
This proper contact time   
is evaluated from the maximum of the right-side $\tau$-hump in 
the logarithmic histogram of figure \ref{fig:distritau}(a). 
\begin{figure}[phtb]
  \centering 
  \includegraphics[width=0.495\textwidth]{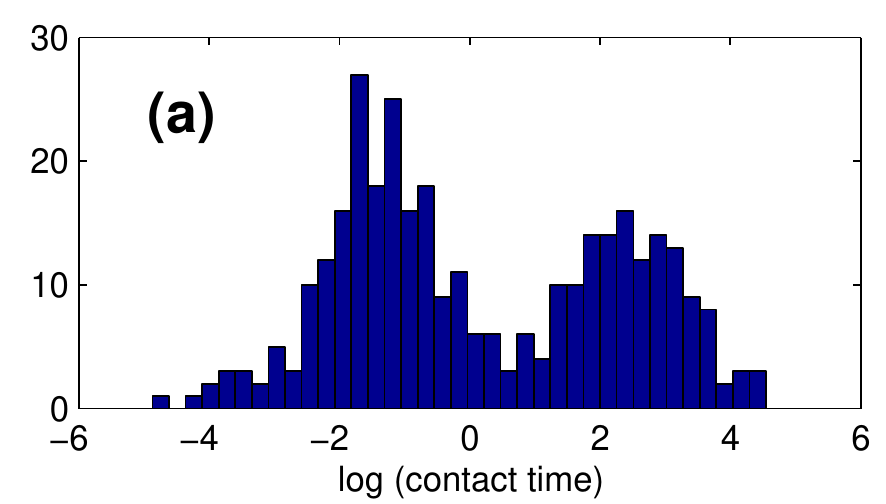} 
  \includegraphics[width=0.495\textwidth]{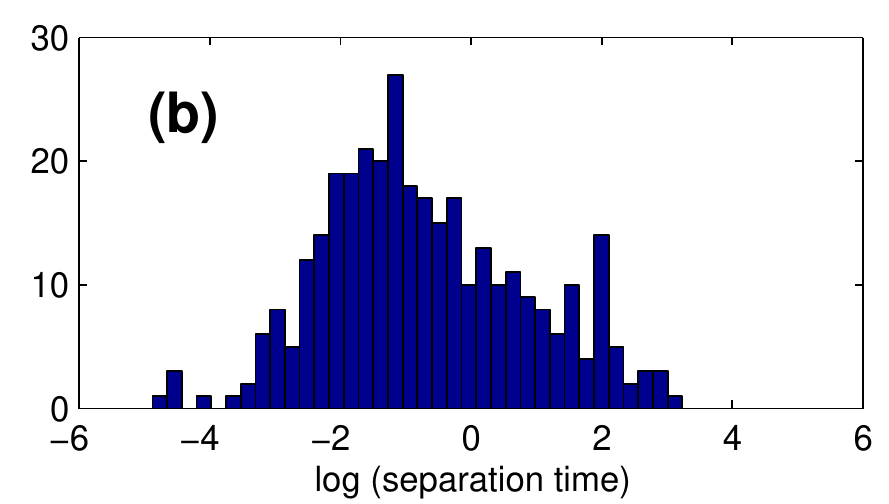}
  \caption{Logarithmically scaled duration histogram of contact 
           intervals (a) and separation intervals (b) for a longer 
           time series ($3000 \, \text{min}$) of the situation in figure 
           \ref{fig:distriP}(a).
  }\label{fig:distritau}.
\end{figure}
The contact states are interrupted by periods of 
faster flickering, as can be seen from the change between 
on and off with a mean duration of $\tau_\text{flicker} \approx
\exp (-1.5) \approx 15 \mathrm s$, which is visible at
the joint lower maximum in both figures \ref{fig:distritau}(a) and (b). 
Otherwise, an intermediate separation of the order of 1-5 min
occurs.

Yet for the second case, $\lambda < 0$, with relatively higher 
locomotion force parameter $\floc$, we observe in figure 
\ref{fig:distriP}(d) dominating periods of flickering 
around the steady state of touching $z_* = 0$, which 
now is a stable deterministic equilibrium, even super-stable for 
$z>0$ with convergence $z(t) \sim ( t_* - t )^{2/3}$ in 
finite time. Longer periods of separation are also evident but 
their distribution does not 
differ much from the situation for $\lambda > 0$ 
(similar to the histogram in figure \ref{fig:distritau}(b), not 
shown here).

\section{Tissue simulations}\label{sec:sim}

After these analytic considerations in the case of cell pair formation
we return to the full equations of motion (\ref{eq:dynamics}). 
Since both time 
and length scale of cell motility processes are well known, the 
only remaining free scaling figure is the magnitude of cell forces. 
In accordance to \cite{Ananthakrishnan2007}, 
here we assume that a typical bundle of 
several actin filaments can exert 
a force of approximately $10 \, \mathrm{pN}$. 
A single cell can, 
with the overall filament density parameter $\tilde \rho$ 
and the force prefactors $\floc, \finth$ 
as in table \ref{tab:parms}, 
reach an effective traction of $\mathcal O (1000 \, \mathrm{pN})$
\begin{table}[htbp]
  \centering 
  \renewcommand{\arraystretch}{1.2}
  \begin{tabular}{|c|c|c|}
    \hline
      $\sqrt{\mcPm} = 3$ & 
      $\tilde \rho  = 9.55 / \mu \mathrm m$ &
      $\tilde \gamma = 2.5 \cdot 10^4 \, \mathrm{pN \, s} / \mu \mathrm m$ \\ 
    \hline
      $\ddd t = 2 \, \mathrm s$ &
      $\Delmi = 0.1$ & 
      $\finth = 60 \, \mathrm{pN}$ \\
      $T      = 120 \, \mathrm s$ & 
      $\Delcr = 0.2 \dots 0.7$ & 
      $\floc  = 10 \dots 20 \, \mathrm{pN}$ \\
    \hline
      \multicolumn{2}{|c|}{
        $b_0 \cdot (\mcPm)^{1/4} = 
          8.31 \, \mathrm{pN} / \sqrt{ \mu \mathrm m \, \mathrm s }$ 
                          } &
      $\alpha = 0 \dots 0.17$ \\
    \hline
  \end{tabular}
  \renewcommand{\arraystretch}{1.0}
  \caption{Simulation and model parameters as described 
           previously. Unless indicated otherwise, all simulations 
           have been performed with this parameter set.
  }\label{tab:parms}
\end{table}
from a force as given by equation (\ref{eq:interaction_force}).
The drag coefficient $\tilde \gamma$ then naturally follows from 
experimentally observed cell velocities \cite{Zahm1997,Friedl1998,Moehl2005}. 
As explained before, $T$ is the persistence time of the intracellular 
cytoskeletal reorganization, and $\mcPm$ determines the relative 
size of the lamella region around the cell body $\mcBri$. 
Moreover, the scaled interaction distances $\Delmi, \Delcr$ as 
defined in equation (\ref{eq:fcyt}) determine the sign 
and scaling of the cell pair interaction force $\vFiji$. 
The relative strength of the vertical component of $\vFiji$ 
in equation (\ref{eq:vffij}) is given by $\alpha$.
Finally, the stochastic perturbation parameter 
$b_0$ in equation (\ref{eq:Wdiscr}) contains  
a factor $(\mcPm)^{-1/4}$ in order to obtain the same amount of 
perturbation for cells with equal body radii $r_i$. 
Since we look for robust features in the simulations, 
$b_0$ was chosen fairly high.

\subsection{Emergence of tissue shape and multiple stable states}

Consider a simple proto-tissue of seven cells as shown in 
\begin{figure}[bthp]
  \centering
  \includegraphics[width=0.98\textwidth]{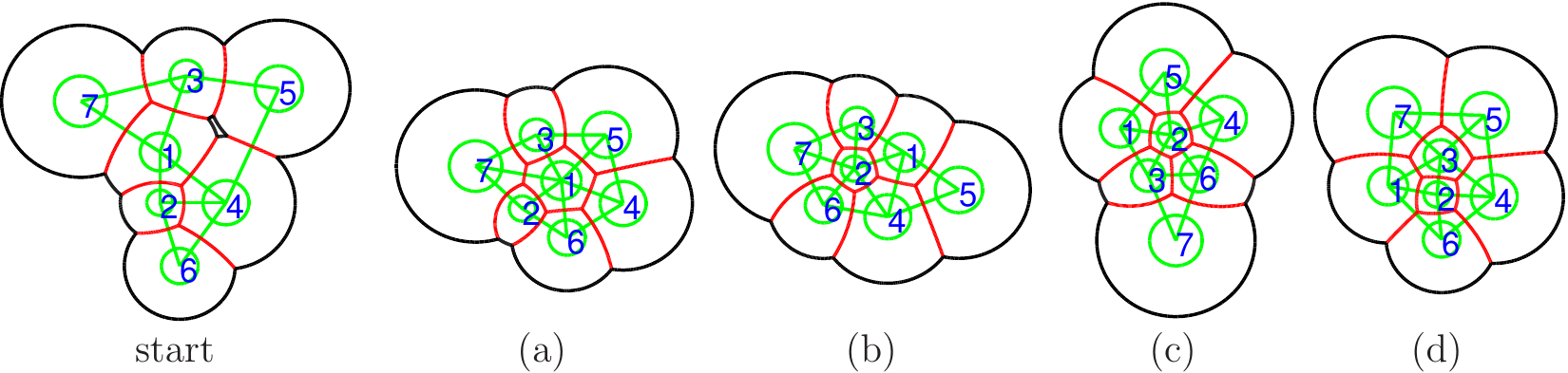} \\
  \vskip 0.03\textwidth
  \includegraphics[width=0.50\textwidth]{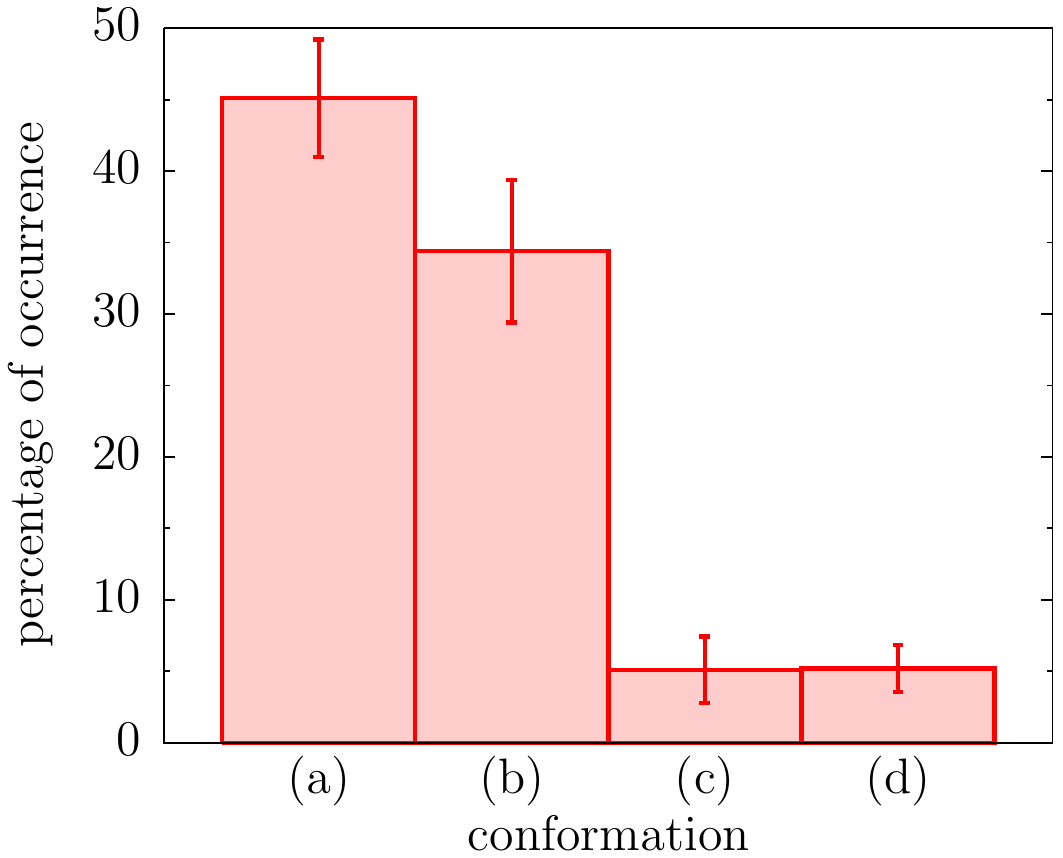}
    \caption{Different tissue conformations (a-d) 
             evolving from the configuration `start' 
             after a simulation time of $8 \, \mathrm h$; 
             here 
             $r_\text{max} = 2.0 \, \mu \mathrm m$ (cell $7$), 
             $r_\text{min} = 1.0 \, \mu \mathrm m$ (cell $2$),
             and parameters 
             $\floc = 20 \, \mathrm{pN}, \; 
             \alpha = 0, \;
             \Delcr = 0.25$.
             The percentage of occurrence 
             of a particular conformation then was computed, 
             and the error bars were obtained by a simple 
             bootstrap method.
             Other features are 
             further explained in the text.
            }\label{fig:tunnel}
\end{figure}
figure \ref{fig:tunnel} `start'. 
Using the parameters 
$\floc = 20 \, \mathrm{pN}, \alpha = 0, \Delcr = 0.25$, 
a series of $1000$ simulations has been performed. 
After a simulation time of $8 \mathrm h$, 
the emerging tissue conformations  
as distinguished by Delaunay network topology 
have been recorded. 
In the course of these $8 \mathrm h$, 
significant changes appear within the tissue, 
and apparently several equilibrium conformations emerge.
For the four most prevalent conformations 
the percentage of occurrence 
is displayed in figure \ref{fig:tunnel}. 
One observes two rather globular shapes (a), (d), 
where either the big cell $1$ or the two small cells $2,3$ 
are engulfed by the others, respectively. 
Furthermore, there are two more elongated shapes (b), (c), 
where only the single small cell $2$ 
is completely surrounded by other cells. 
Being distinguished by topology, 
(b) and (c) are in fact quite close in shape, 
despite of their rotational variation. 

From the high occurrence of the topological conformations (a), (b) 
one might conclude that these two conformations are the most stable ones. 
Thus, and 
to clarify the interrelations between the 
conformations (a-d), we investigate (a) and (b) in longer simulations. 
To this end, by starting from the configurations (a) and (b) 
(see figure \ref{fig:tunnel}), both tissues have been evolved 
for $40$ additional hours of simulation time. 
In order to characterize the shape of tissue
with respect to global and elongated shape, 
here we observe \emph{tissue size}, \ie the maximum diameter, 
and \emph{tissue circularity} 
\begin{equation*}
  \Omega = \frac{ 2 \sqrt{ \pi A_\text{tiss} } }{ \sum_i |\Gamma_{i0}| } 
         \leq 1, 
\end{equation*}
where $A_\text{tiss}$ is the total area of the tissue. 
Note that in connected tissues $\Omega = 1$ would be attained 
for a purely circular globe. 
Other observables are possible but less indicative 
in this context. 
\begin{figure}[tbh]
  \centering
  \includegraphics[width=0.98\textwidth]{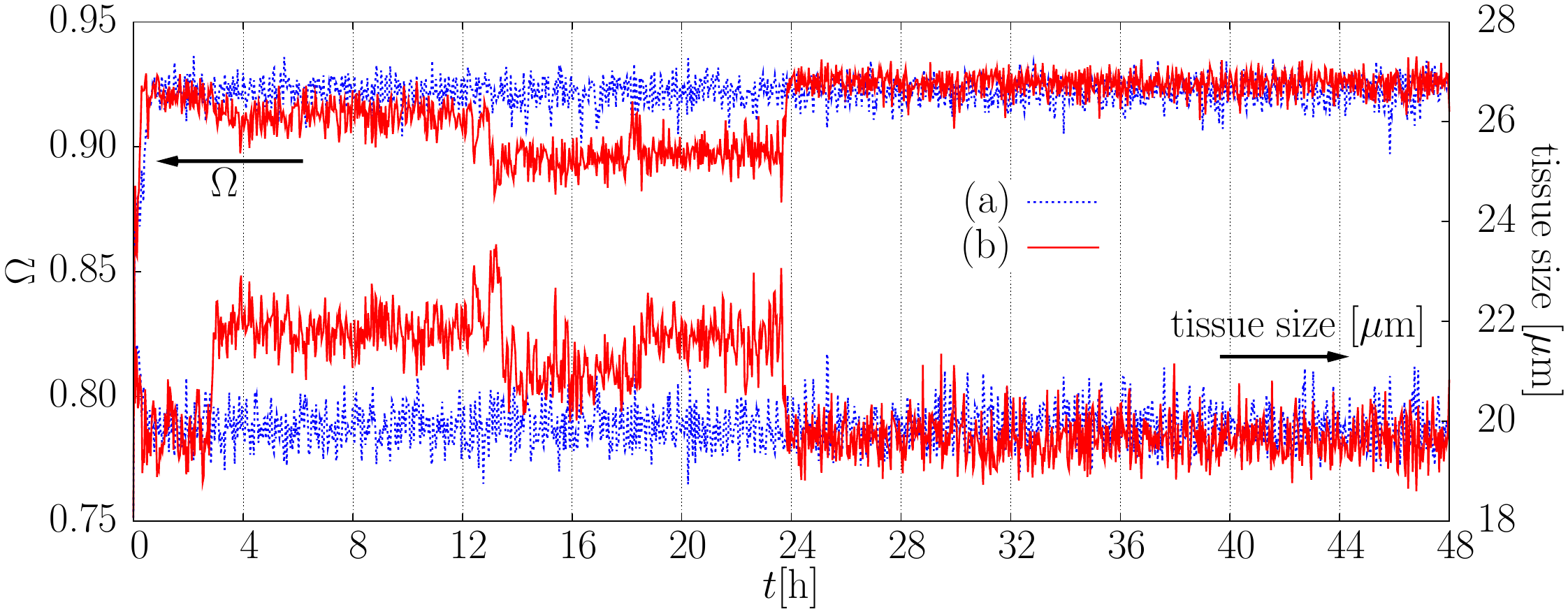}
  \caption{Two time series (blue, red) for the tissue 
           in figure \ref{fig:tunnel} `start'. 
           After $8 \, \mathrm h$ the two distinct 
           topological conformations 
           from figure \ref{fig:tunnel}(a) (blue) 
           and \ref{fig:tunnel}(b) (red) have emerged. 
           These topological conformations are characterized by 
           distinct tissue size (right axis) and 
           circularity $\Omega$ (left axis). 
           While (a) is apparently a stable topological conformation 
           that does not change even for strong stochastic perturbations, 
           (b) relaxes into a topological conformation of globular shape  
           via several intermediate steps, see supplementary 
\href{http://www.theobio.uni-bonn.de/people/mab/sup01/index.html}{\texttt{mova.avi}}, 
\href{http://www.theobio.uni-bonn.de/people/mab/sup01/index.html}{\texttt{movb.avi}}.
          }\label{fig:tunneltime}
\end{figure}
In figure \ref{fig:tunneltime}, 
the two time series (blue, red) soon reach the different conformations 
of figure \ref{fig:tunnel}(a) and (b) respectively, 
at times around $t \sim 8 \mathrm h$.
While the circularity $\Omega$ is 
only slightly different in the two cases, 
the tissue size is clearly higher for the elongated conformation (b).
Apart from stochastic fluctuations and an initial 
equilibration phase for $t < 1 \, \mathrm h$, 
both observables attain a constant value for time series (a). 
In contrast, for time series (b) there appear distinct states between 
$t \sim 2.5 \, \mathrm h$ and $t \sim 24 \, \mathrm h$. 
Indeed these observations are reflected by the actual evolution 
of the tissue. While the topology of 
the tissue does not change after $t = 2 \, \mathrm h$ for time series (a) 
(see supplementary \href{http://www.theobio.uni-bonn.de/people/mab/sup01/index.html}{\texttt{mova.avi}}), 
tissue (b) 
(\href{http://www.theobio.uni-bonn.de/people/mab/sup01/index.html}{\texttt{movb.avi}})
goes through several different conformational states. 
At $t=13.5 \, \mathrm h$ it attains the same topology as conformation (c), 
identifying (c) as a transient state
(\href{http://www.theobio.uni-bonn.de/people/mab/sup01/index.html}{\texttt{movc.avi}}). 
Afterwards, approximately at 
$t \sim 18.5 \, \mathrm h$, cell $7$ establishes contact with cells $1,3$, 
so that the tissue shape is similar to (c). Finally, shortly before 
$t = 24 \, \mathrm h$, the tissue reaches its final conformation 
similar to (d) except for the order of the marginal cells. 
Conformation (d) emerges in a similar manner as (a), 
however instead of cells $1,5$ initially cells $3,4$ form 
a neighbor pair, quickly leading to the stable final arrangement 
in less than $0.5 \, \mathrm h$ (see supplementary 
\href{http://www.theobio.uni-bonn.de/people/mab/sup01/index.html}{\texttt{movd.avi}}).
Moreover, the time series 
of conformation (b) in figure \ref{fig:tunneltime} 
suggests, that during $t = 0.5 \dots 3 \, \mathrm h$ the tissue 
already attains a shape of similar compactness and stability 
as in figure \ref{fig:tunnel}(a). Nevertheless, the Delaunay mesh (green lines) 
is not convex there 
(\href{http://www.theobio.uni-bonn.de/people/mab/sup01/index.html}{\texttt{movb.avi}}),
which explains this surprising instability.

It appears, that the stability of a tissue is related 
to its globular shape. This is not a surprise, 
since the stochastic forces in equation (\ref{eq:free_boundary_force}) 
are defined only on free cell boundaries $\Gamma_{i0}$, 
and therefore act only on marginal cells. 
Thus, by minimizing the extent of all $\Gamma_{i0}$,
a maximal circularity minimizes stochastic perturbations, 
which enhances the stability of the tissue. 
Additionally, for a tissue to change its topology, 
its cells have to overcome barriers as imposed by the
other cells. For example, cell $3$ has to displace 
cell $1$ in   
\href{http://www.theobio.uni-bonn.de/people/mab/sup01/index.html}{\texttt{movb.avi}} 
at $t \approx 3 \, \mathrm h$ in order to make contact with $2$. 
Depending on the particular configuration, the severity of 
these barriers might range from prohibitive to practically 
non-existent. 
Influenced by the strength of the stochastic 
interactions, these barriers then determine the time 
scale of further relaxation to equilibrium. 
In this sense, the notion of equilibrium is directly 
related to an inherent time scale. 
According to the 
previous evolution of the tissue, there may be 
several stable topological conformations for a given time scale.

\subsection{Stability of tissue formation}

In order to explore the ramifications of piecewise spherical cells 
within our model framework, we study the influence of 
$\mcPm$ and $\Delcr$ on tissue formation.
To this end, a simulation has been performed starting 
from an exemplary configuration as in figure \ref{fig:vor_closure} 
with $\Delcr = 0.3, \sqrt{\mcPm} = 3,\alpha = 0.17$ 
and $\floc = 10 \, \mathrm{pN}$. 
After $8 \, \mathrm h$, either $\Delcr$ or $\sqrt{\mcPm}$ was 
modified to a nearby parameter position as indicated 
in figure \ref{fig:pmaxtyp}, 
\begin{figure}[tbhp]
  \centering
  \includegraphics[width=0.98\textwidth]{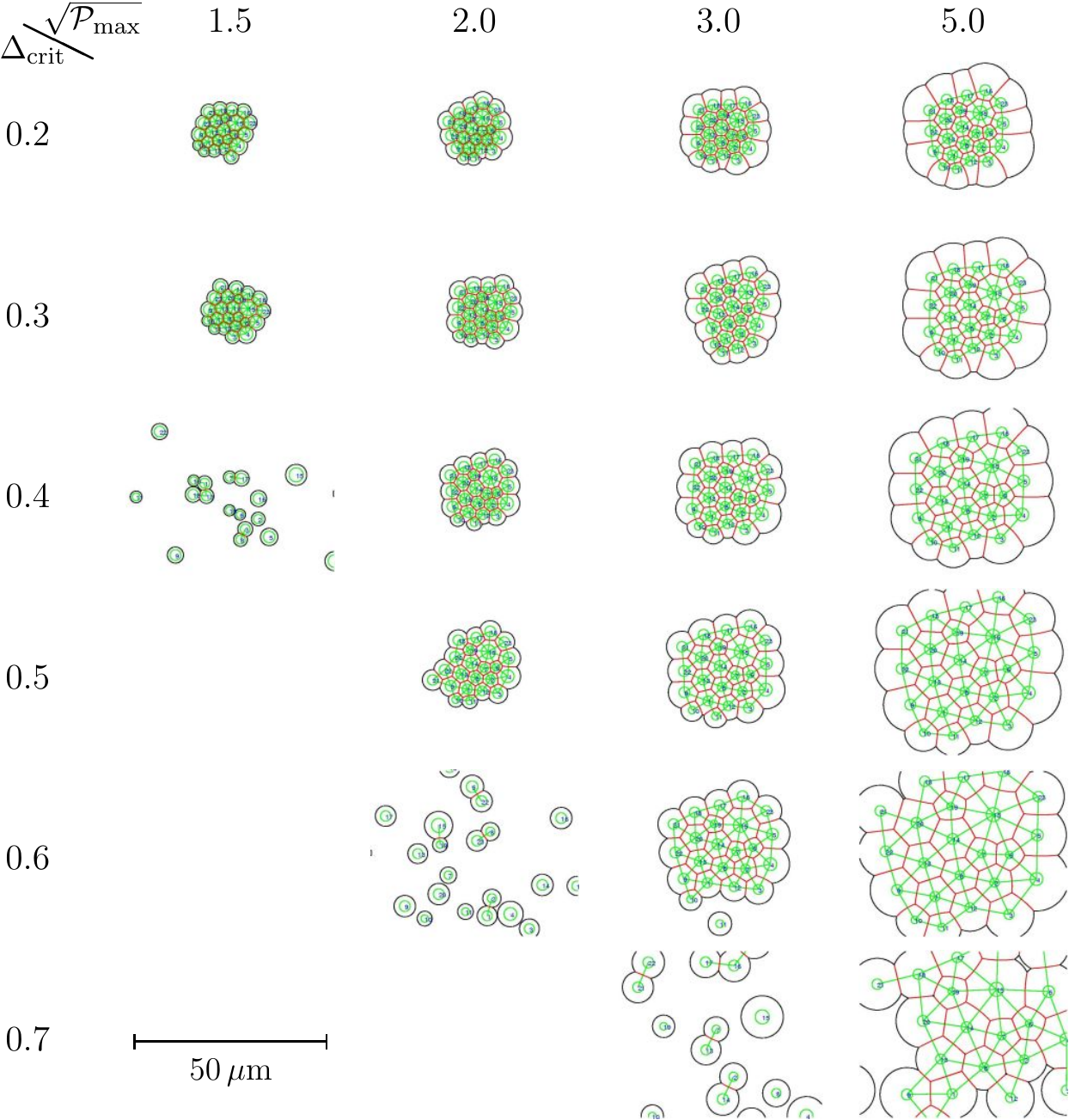}
  \caption{Stability of tissue for various values of the parameters 
           $\sqrt{\mcPm}$ and $\Delcr$, where $\floc = 10 \, \mathrm{pN}$, 
           and $\alpha = 0.17$.
           When increased, both parameters 
           lead to an increased tissue size. For sufficiently large 
           $\Delcr$, the tissue eventually dissociates. 
           The extremal cell body radii are 
           $r_\mathrm{min} = 0.9 \, \mu \mathrm m$ and 
           $r_\mathrm{max} = 1.7 \, \mu \mathrm m$ 
           in all simulations presented in this figure.
  }\label{fig:pmaxtyp}
\end{figure}
and the simulation was continued for further $8 \, \mathrm h$. 
This procedure was repeated until the whole panel in 
\ref{fig:pmaxtyp} was filled with the final tissue configurations.

For fixed $\{ r_i \}$, the overall size of the tissue is 
determined by both $\sqrt{\mcPm}$, defining the free cell 
radius $\Rif$ in units of $r_i$, and $\Delcr$, presetting the 
equilibrium cell-cell body distance in units of $\deltru$, 
see equations (\ref{eq:Rif}), (\ref{eq:fcyt}) 
and (\ref{eq:deltru}). Correspondingly, in 
figure \ref{fig:pmaxtyp}, the overall tissue extension 
increases from left to right and from top to bottom. 
Furthermore, for given $\Delcr$, tissues with higher $\sqrt{\mcPm}$ 
exhibit a rather compact, almost quadratic shape. 
We speculate that this is due to spontaneous formation of 
distinct protrusions arising from stochastic perturbations 
and leading to an increase of locomotion at the corners. 
In contrast, tissues with lower $\sqrt{\mcPm}$ 
feature more irregular margins. 
Similarly, for given $\sqrt{\mcPm}$, 
larger values of $\Delcr$ yield more irregularity, 
most pronounced directly before dissociation of the tissue.

Apparently, the emerging interaction forces 
are sufficiently strong for tissue aggregation only if there is  
enough space for the adaptation of neighboring cell lamellae. 
This space serves as a cushion for accommodating near-range 
repulsion from multiple neighbor cells and at the same time 
poses a resistance to stochastic perturbations 
by mid-range neighbor cell attraction, see figure \ref{fig:deltfloc}. 
Otherwise the tissue dissociates, leading to isolated 
cells exclusively driven by stochastic perturbations. 
In order to quantify these findings, consider the 
relative lamella width $\lamwid$ and the dimensionless   
cell overlap $\covl$ defined by 
\begin{equation*}
  \lamwid = \frac{ \smcPm - 1 }{ \smcPm } < 1, \qquad \qquad
  \covl = 1 - \Delcr < 1.
\end{equation*}
Inserting the marginal values $\smcPm$ and $\Delcr$ of 
those tissues in figure \ref{fig:pmaxtyp} that are not 
yet dissociated, one recognizes that the product 
$\covl \cdot \lamwid =: \sprod$ is approximately constant, 
with $\sprod \approx 0.24$, 
\begin{figure}
  \centering
  \includegraphics[width=0.60\textwidth]{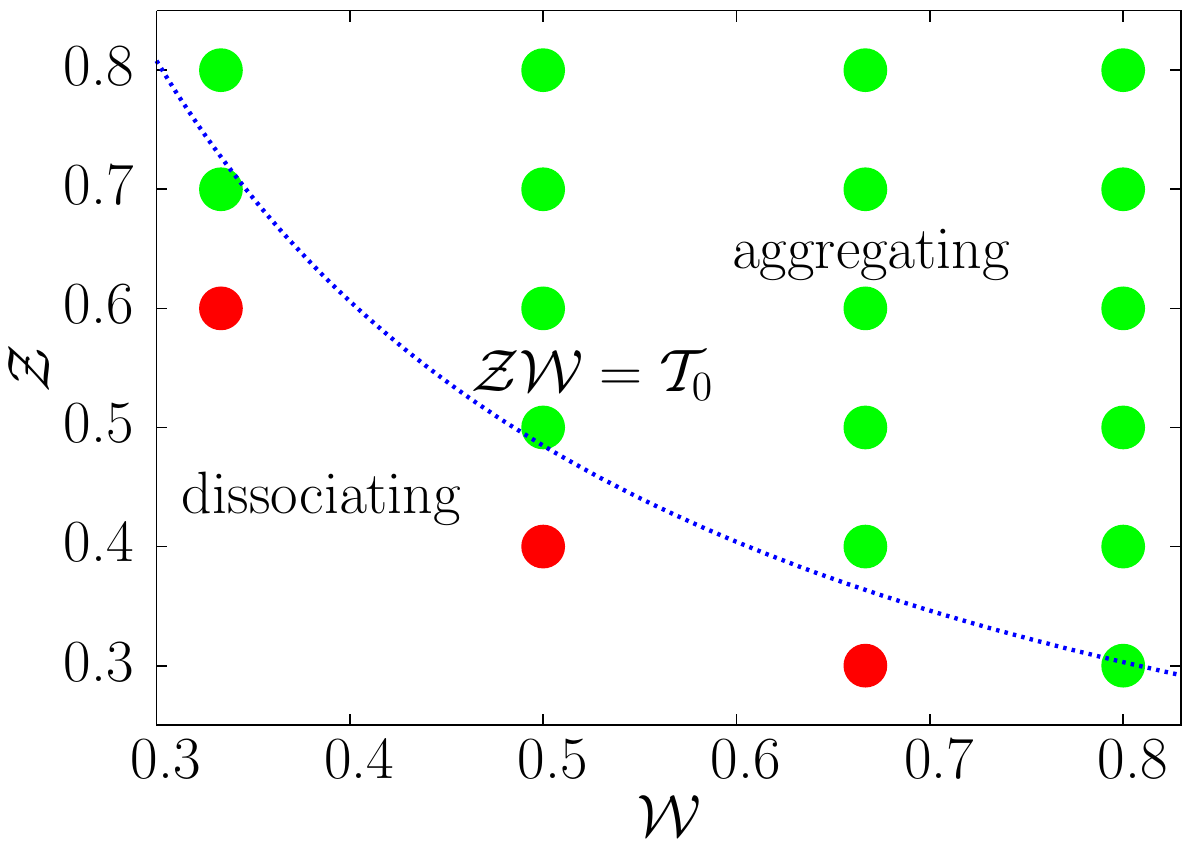} 
    \caption{Tissue aggregation and dissociation 
             depending on the relative lamella width $\lamwid$ 
             and the dimensionless cell overlap $\covl$. 
             The line 
             $\covl \cdot \lamwid = \sprod \approx 0.24 \pm 3\% $ 
             separates associating from dissociating tissues 
             and was determined from a linear regression
             against $\covl = \sprod / \lamwid$. See also 
             the explanations in the text.
            }\label{fig:sprod}
\end{figure}
see figure \ref{fig:sprod}. 
This $\sprod$ can be identified as a threshold value 
guaranteeing tissue coherence under the condition 
\begin{equation}\label{eq:condsprod}
  \covl \cdot \lamwid \geq \sprod, 
\end{equation}
where $\sprod$ eventually depends on the other parameters, 
which were fixed here.
This confirms that for tissue aggregation to occur, 
$\smcPm$ and with it the relative lamella width $\lamwid$ 
must be sufficiently large.

On the other hand, we have established this result 
under the tissue homogeneity condition (\ref{eq:pmax_max}) 
guaranteeing starlikeness of cells, which now can be rewritten 
as 
\begin{equation}\label{eq:starlike2}
  \covl_\text{max} \cdot \lamwid = 1 - \frac{1}{\qnb}, 
\end{equation}
with $\covl_\text{max} := 1 - \Delmi$ defining the 
maximal dimensionless overlap. 
Since $\lamwid < 1$ by construction, 
inequality (\ref{eq:starlike2}) always holds 
for very high cell size homogeneities $\qnb \geq 1/\Delmi$. 
However, for lower $\qnb$ there is an upper bound on $\lamwid$ 
restricting the available space for the cell lamellae. 
If, in addition, starlikeness of cells is enforced 
for all possible neighborhood constellations, 
then $\qnb$ has to be replaced by $Q$, 
see equation (\ref{eq:pmax_max1}). 
In this way the relations (\ref{eq:condsprod}) and (\ref{eq:starlike2}) 
lead to the sufficient condition for tissue coherence 
\begin{equation}\label{eq:starlike3}
  \frac{ \sprod }{ 1 - \Delcr } \leq 
  \lamwid = 1 - \frac{ 1 }{ \smcPm } \leq
  \frac{ 2 }{ (1-\Delmi)(1+r_\text{max}/r_\text{min}) }.
\end{equation}
From these estimates we finally conclude that 
for given model force parameters 
($\Delmi$, $\Delcr$, $\finth$, $\alpha$, $\floc$, $b_0$) 
the 
formation of integer tissue aggregates 
with overall starlike cells is guaranteed within a certain 
finite range of the free cell size parameter $\smcPm$,  
where the upper bound decreases with an increasing ratio 
$r_\text{max} / r_\text{min}$ of extremal cell body radii.
Within the limits of inequality (\ref{eq:starlike3}) 
the lamellae regions are wide enough 
to perform the necessary deformations by 
adapting to the surrounding neighbors through shape changes. 
Thus, nature's freedom in developing aggregating tissues 
may be constrained by a tradeoff between the relative size 
of cells with respect to their bodies ($\smcPm$) and 
the cell size heterogeneity ($1/Q$).

\section{Results and discussion}\label{sec:results}

In this article we have investigated the emergence of tissue 
aggregation using finite, generalized Voronoi neighborhoods  
as a basis for the description of cells within epithelial tissues. 
It was shown that the two-dimensional 
geometric structure observed in tissues, in particular the 
circular shape of contact arcs and the  
size distribution of cells can be captured 
by the experimentally accessible characteristics 
of cell body radii $\{ r_i \}$ 
and the relative extension $\mcPm$ of the lamellae.
The quotient method defined 
by equation (\ref{eq:ratrule_part}) with unique weight factors $w_i = r_i$ 
implements the expected partition of influence regions and 
additionally suggests the definition of directed force 
densities on the contact and free boundaries. 
In contrast, the difference method is not appropriate, 
because the influence region of a cell growing in size 
becomes reduced within a cell pair.
Nevertheless this method has been used previously for modelling 
three-dimensional tissue dynamics \cite{SchallerThesis,Schaller2005}. 
By using the particular additive weights $w_i = r_i$, 
the specific geometrical properties (``orthocircles'') 
allow for a dense regular triangulation of space between 
neighboring cell bodies even in the case of overlap, 
see \cite{Beyer2007} for details. 

In the asymptotic case of two cells in brief contact, \ie
$\deltij \approx \deltru$, the approximate equations of 
overdamped motion without persistence were treated   
analytically. Thereby, the contact singularity at 
the rupture points could be resolved and a stationary 
probability distribution was found under the assumption 
that separated cells indirectly attract each other due 
to biased locomotion.
Corresponding simulations confirm that in the case of dominating 
attraction, \ie contact parameter $\lambda > 0$, there is indeed a unique, 
stable equilibrium position facilitating cell-cell attachment 
after the first contact. 

The simulations of multicellular dynamics, \cf figure \ref{fig:pmaxtyp},
reveal a surprising richness of tissue structure, 
with emerging shapes from elongated to globular and 
from compact quadratic to irregular and dispersed. 
Ranging from widely spread to tightly contracted, 
different lamella protrusions appear for varying $\mcPm$. 
Similarly, within a more or less compressed tissue 
as in figure \ref{fig:vor_closure}, 
single cells attain varying forms 
from almost circular at the margins to polygonal in the interior.
Moreover, distinct stresses throughout the tissue  
lead to both narrowly compressed or widely stretched 
cell shapes. 

Furthermore, due to the stochastic nature of 
neighbor constellation and cell-cell interactions, 
different stable tissue configurations may evolve from 
the same starting configuration. The underlying dynamical 
time scales vary over several orders of magnitude. 
This can be interpreted in terms of an incomplete balance 
of competing interactions within a system. 
In the course of relaxation to equilibrium, this 
is known to cause \emph{frustration}, meaning that the system ends up in a 
local minimum of the involved (generalized) free energy 
\cite{SemmrichKroy2007,Purnomo2008,Janke2008}. 
Thus, the structure of the proposed anisotropic and active 
force interactions exhibits interesting non-trivial features, 
yet it is sufficiently elementary to allow for 
a rigorous treatment under certain further assumptions.

Finally, 
the limits of aggregation due to a prohibitively small 
tissue coherence threshold $\sprod$ have been explored. 
This leads to our main result 
following from purely geometric arguments: 
In proposition \ref{prop:starlike} we have derived that 
relatively large lamellae extensions $\mcPm$ 
are restricted by the 
cell size homogeneity $\qnb$ when requiring starlikeness of cells. 
On biological grounds one might argue that not 
all observed cells are starlike (in the mathematical sense). 
In such a case, however, the exertion of forces onto neighboring cells is 
certainly hindered -- especially in cell regions 
that cannot be reached by radial filaments. There, the necessary 
centro-radial support of the apical actin cortex 
may be weakened due to the necessary bending of 
filament bundles.
Reversing the argument, for the cell to provide 
macroscopic stability within the tissue, it is 
beneficial to attain a starlike shape.

Several model refinements are of interest. 
On the cell biological side, cell division is a  
commonly observed phenomenon, having special consequences 
for pattern formation and self-organization during 
embryonic development. 
In this direction also the growth 
of cells, or even the complete cell cycle including 
necrosis or apoptosis could be considered. 
On the mathematical side 
more elaborate stochastics, such as differentially modeled 
stochastic forces on the free boundaries as well as on the 
contact borders could be implemented. 
In particular, in the light of recent experimental results 
\cite{Sivaramakrishnan2008,Koestler2008}, 
varying filament orientations or 
binding and unbinding processes of linker molecules 
might be considered. 
Moreover, we emphasize that a three-dimensional generalization 
of the multiplicatively weighted Voronoi tessellation is straight-forward, 
opening a wide potential for applications, for example 
epidermal tissue organization in the intestinal crypt 
without artificially imposed constraints.
Finally, explicit variables for cell polarization and 
reorganization of the cytoskeleton could possibly 
lead to refined dynamics closer resembling the 
behavior observed \emph{in vivo}.

\textbf{Acknowledgements:}
We thank D.~B\"ar, G.~Schaller and M.~Meyer-Hermann for helpful discussions, 
and G.~Wenzel and G.~Kirfel for providing us with microscopic pictures. 
This work was supported by the
\foreignlanguage{german}{Deutsche Forschungsgemeinschaft, SFB 611}  
``Singular Phenomena and Scaling in Mathematical Models''.

\flushleft


\begin{thebibliography}{10}

  \bibitem{Alberts}
  B.~Alberts, A.~Johnson, J.~Lewis, K.~Roberts, and P.~Walter, editors.
  \newblock {\em Molecular {B}iology of the {C}ell}.
  \newblock Garland Science, 4th edition, 2002.
  \newblock Chapters 16 and 19.

  \bibitem{AltAlt2008}
  H.-W.~Alt and W.~Alt.
  \newblock Phase boundary dynamics: Transitions between ordered and disordered
    lipid monolayers.
  \newblock {\em Interfaces and Free Boundaries}, 11:1, 2009.

  \bibitem{Alt2003}
  W.~Alt.
  \newblock Nonlinear hyperbolic systems of generalized {Navier-Stokes} type for
    interactive motion in biology.
  \newblock In: S.~Hildebrandt and H.~Karcher, editors, {\em Geometric Analysis
    and Nonlinear Partial Differential Equations}, page 431. Springer, 2003.

  \bibitem{Ananthakrishnan2007}
  R.~Ananthakrishnan and A.~Ehrlicher.
  \newblock The forces behind cell movement.
  \newblock {\em International Journal of Biological Sciences}, 3:303, 2007.

  \bibitem{Arnold1974}
  L.~Arnold.
  \newblock {\em Stochastic differential equations: theory and applications}.
  \newblock Wiley Interscience, 1974.

  \bibitem{Ash1986}
  P.~Ash and E.~Bolker.
  \newblock {Generalized Dirichlet tessellations}.
  \newblock {\em {Geometrica Dedicata}}, 20:209, 1986.

  \bibitem{Aurenhammer1983}
  F.~Aurenhammer and H.~Edelsbrunner.
  \newblock An optimal algorithm for constructing the weighted 
    {V}oronoi diagram in the plane.
  \newblock {\em Pattern Recognition}, 17:251, 1984.

  \bibitem{AurenhammerKlein1996}
  F.~Aurenhammer and R.~Klein.
  \newblock Voronoi {D}iagrams.
  \newblock Technical Report 198, Fern{U}niversit\"at Hagen, 1996.
  \newblock \url{http://wwwpi6.fernuni-hagen.de/Publikationen/tr198.pdf}.

  \bibitem{Bernal1992}
  J.~Bernal.
  \newblock Bibliographic notes on {V}oronoi diagrams.
  \newblock Technical Report 5164, {U.S.} Dept.~of {C}ommerce, National
    {I}nstitute of {S}tandards and {T}echnology, 1993.
  \newblock \url{ftp://math.nist.gov/pub/bernal/or.ps.Z}.

  \bibitem{Beyer2007}
  T.~Beyer and M.~Meyer-Hermann.
  \newblock Modeling emergent tissue organization involving high-speed migrating
    cells in a flow equilibrium.
  \newblock {\em Physical Review E}, 76:021929, 2007.

  \bibitem{Beyer2005}
  T.~Beyer, G.~Schaller, A.~Deutsch, and M.~Meyer-Hermann.
  \newblock Parallel dynamic and kinetic regular triangulation in three
    dimensions.
  \newblock {\em {Computer Physics Communications}}, 172:86, 2005.

  \bibitem{Brevier2008}
  J.~Brevier, D.~Montero, T.~Svitkina, and D.~Riveline.
  \newblock The asymmetric self-assembly mechanism of adherens junctions: A
    cellular push--pull unit.
  \newblock {\em Physical Biology}, 5:016005, 2008.

  \bibitem{Brillouin1930}
  L.~Brillouin.
  \newblock Les \'electrons dans les m\'etaux et le classement des ondes de de
    {B}roglie correspondantes.
  \newblock {\em Comptes {R}endus {H}ebdomadaires des {S}\'eances de
    l'{A}cad\'emie des {S}ciences}, 191:292, 1930.

  \bibitem{Brodland2002}
  W.~Brodland and J.~Veldhuis.
  \newblock Computer simulations of mitosis and interdependencies between 
    mitosis orientation, cell shape and epithelia reshaping.
  \newblock {\em {Journal of Biomechanics}}, 35:673, 2002.

  \bibitem{Dieterich2004}
  P.~Dieterich, J.~Seebach, and H.~Schnittler.
  \newblock Quantification of shear stress-induced cell migration in endothelial
    cultures.
  \newblock In: A.~Deutsch, M.~Falcke, J.~Howard and W.~Zimmermann, editors, 
    {\em {Function and Regulation of Cellular Systems: Experiments and Models}},
    {Mathematics and Biosciences in Interaction}, page 199. Birkh\"auser, 2004.

  \bibitem{Dirichlet1850}
  G.L.~Dirichlet.
  \newblock {\"U}ber die {R}eduction der positiven quadratischen {F}ormen mit
    drei unbestimmten ganzen {Z}ahlen.
  \newblock {\em Journal f\"ur Reine und Angewandte Mathematik}, 40:209, 1850.

  \bibitem{DrasdoForgacs2000}
  D.~Drasdo and G.~Forgacs.
  \newblock Modeling the interplay of generic and genetic mechanisms in 
    cleavage, blastulation, and gastrulation.
  \newblock {\em Developmental {D}ynamics}, 219:182, 2000.

  \bibitem{Drasdo1995}
  D.~Drasdo, R.~Kree, and J.S.~McCaskill.
  \newblock Monte {C}arlo approach to tissue-cell populations.
  \newblock {\em Physical Review E}, 52:6635, 1995.

  \bibitem{Evans1997}
  E.~Evans and K.~Ritchie.
  \newblock {Dynamic strength of molecular adhesion bonds.}
  \newblock {\em Biophysical Journal}, 72:1541, 1997.

  \bibitem{Friedl1998}
  P.~Friedl, K.S.~Z\"anker, and E.-B.~Br\"ocker.
  \newblock Cell migration strategies in {3-D} extracellular matrix: Differences
    in morphology, cell matrix interactions and integrin function.
  \newblock {\em Microscopy Research and Technique}, 43:369, 1998.

  \bibitem{Galle2005}
  J.~Galle, M.~Loeffler, and D.~Drasdo.
  \newblock Modeling the effect of deregulated proliferation and apoptosis 
    on the growth dynamics of epithelial cell populations in vitro.
  \newblock {\em Biophysical Journal}, 88:62, 2005.

  \bibitem{Gambin2006}
  Y.~Gambin, R.~Lopez-Esparza, M.~Reffay, E.~Sierecki, N.S.~Gov, M.~Genest,
    R.S.~Hodges, and W.~Urbach.
  \newblock Lateral mobility of proteins in liquid membranes revisited.
  \newblock {\em Proceedings of the National Academy of Sciences, USA}, 
    103:2098, 2006.

  \bibitem{Hegerfeldt2002}
  Y.~Hegerfeldt, M.~Tusch, E.-B.~Br\"ocker, and P.~Friedl.
  \newblock Collective cell movement in primary melanoma explants: Plasticity of
    cell-cell interaction, $\beta 1$-integrin function and migration strategies.
  \newblock {\em Cancer Research}, 62:2125, 2002.

  \bibitem{Honda1978}
  H.~Honda.
  \newblock Description of cellular patterns by {D}irichlet domains: The
    two-dimensional case.
  \newblock {\em Journal of Theoretical Biology}, 72:523, 1978.

  \bibitem{Honda2004}
  H.~Honda, M.~Tanemura, and T.~Nagai.
  \newblock A three-dimensional vertex dynamics cell model of space-filling
    polyhedra simulating cell behavior in a cell aggregate.
  \newblock {\em Journal of Theoretical Biology}, 226:439, 2004.

  \bibitem{Janke2008}
  W.~Janke, editor.
  \newblock {\em {Rugged Free Energy Landscapes: Common Computational Approaches
    to Spin Glasses, Structural Glasses and Biological Macromolecules}}, volume
    736 of {\em {Lecture Notes in Physics}}.
  \newblock Springer, Berlin, 2008.

  \bibitem{Kirfel2003}
  G.~Kirfel, A.~Rigort, B.~Borm, C.~Schulte, and V.~Herzog.
  \newblock Structural and compositional analysis of the keratinocyte migration
    track.
  \newblock {\em Cell Motility and the Cytoskeleton}, 55:1, 2003.

  \bibitem{Kloeden1999}
  P.E.~Kloeden, E.~Platen.
  \newblock {\em Numerical solution of stochastic differential equations}.
  \newblock Springer, 1992.
  \newblock Chapter 8.

  \bibitem{Koestler2008}
  S.A.~Koestler, S.~Auinger, M.~Vinzenz, K.~Rottner, and J.V.~Small.
  \newblock Differentially oriented populations of actin filaments generated in
    lamellipodia collaborate in pushing and pausing at the cell front.
  \newblock {\em Nature Cell Biology}, 10:306, 2008.

  \bibitem{KuuselaAlt2009}
  E.~Kuusela and W.~Alt.
  \newblock Continuum model of cell adhesion and migration.
  \newblock {\em Journal of Mathematical Biology}, 58:135, 2009.

  \bibitem{Marie2003}
  H.~Marie, S.J.~Pratt, M.~Betson, H.~Epple, J.T.~Kittler, L.~Meek, 
    S.J.~Moss, S.~Troyanovsky, D.~Attwell, G.D.~Longmore, and V.M.~Braga.
  \newblock The {LIM} protein {A}juba is recruited to cadherin-dependent cell
    junctions through an association with alpha-catenin.
  \newblock {\em Journal of Biological Chemistry}, 278:1220, 2003.

  \bibitem{Meineke2001}
  F.~Meineke, S.~Potten, and M.~Loeffler.
  \newblock Cell migration and organization in the intestinal crypt using a
    lattice-free model.
  \newblock {\em {Cell Proliferation}}, 34:253, 2001.

  \bibitem{Moehl2005}
  C.~M\"ohl.
  \newblock \foreignlanguage{german}{Modellierung von Adh\"asions- und
    Cytoskelett-Dynamik in Lamellipodien migratorischer Zellen}.
  \newblock Diploma thesis, \foreignlanguage{german}{Universit\"at Bonn}, 2005.

  \bibitem{Purnomo2008}
  E.H.~Purnomo, D.~van~den Ende, S.A.~Vanapalli, and F.~Mugele.
  \newblock Glass transition and aging in dense suspensions of thermosensitive
    microgel particles.
  \newblock {\em Physical Review Letters}, 101:238301, 2008.

  \bibitem{SchallerThesis}
  G.~Schaller.
  \newblock {\em On selected numerical approaches to cellular tissue}.
  \newblock PhD thesis, Johann Wolfgang Goethe-Universit\"at, Frankfurt am Main,
    2005.

  \bibitem{Schaller2004}
  G.~Schaller and M.~Meyer-Hermann.
  \newblock {Kinetic and dynamic Delaunay tetrahedralizations in three
    dimensions}.
  \newblock {\em {Computer Physics Communications}}, 162:9, 2004.

  \bibitem{Schaller2005}
  G.~Schaller and M.~Meyer-Hermann.
  \newblock Multicellular tumor spheroid in an off-lattice {V}oronoi-{D}elaunay
    cell model.
  \newblock {\em Physical Review E}, 71:051910, 2005.

  \bibitem{SemmrichKroy2007}
  C.~Semmrich, T.~Storz, J.~Glaser, R.~Merkel, A.R.~Bausch, and K.~Kroy.
  \newblock Glass transition and rheological redundancy in {F}-actin solutions.
  \newblock {\em Proceedings of the National Academy of Sciences, USA}, 
    104:20199, 2007.

  \bibitem{Shamos1975}
  M.~Shamos and D.~Hoey.
  \newblock Closest point problems.
  \newblock In {\em Proceedings of the 16th Annual IEEE Symposium on Foundations
    of Computer Science (FOCS)}, page 151, 1975.

  \bibitem{Sivaramakrishnan2008}
  S.~Sivaramakrishnan, J.V.~DeGuilio, 
    L.~Lorand, R.D.~Goldman, and K.M.~Ridge.
  \newblock Micromechanical properties of keratin 
    intermediate filament networks.
  \newblock {\em Proceedings of the National Academy of Sciences, USA}, 
    105:889, 2008.

  \bibitem{Sulsky1984}
  D.~Sulsky, S.~Childress, and J.K.~Percus.
  \newblock A model for cell sorting.
  \newblock {\em Journal of Theoretical Biology}, 106:275, 1984.

  \bibitem{Taute2008}
  K.M.~Taute, F.~Pampaloni, E.~Frey, and E.-L.~Florin.
  \newblock Microtubule dynamics depart from the wormlike chain model.
  \newblock {\em Physical Review Letters}, 100:028102, 2008.

  \bibitem{Thiessen1911}
  A.H.~Thiessen.
  \newblock Precipitation averages for large areas.
  \newblock {\em Monthly Weather Reviev}, 39:1082, 1911.

  \bibitem{Tinkle2008}
  C.L.~Tinkle, A.~Pasolli, N.~Stokes, and E.~Fuchs.
  \newblock New insights into cadherin function in epidermal sheet formation and
    maintenance of tissue integrity.
  \newblock {\em Proceedings of the National Academy of Sciences, USA}, 105:15405, 2008.

  \bibitem{Voronoi1908}
  G.~Voronoi.
  \newblock Nouvelles applications des param\`etres continus \`a la th\'eorie de
    formes quadratiques.
  \newblock {\em Journal f\"ur die Reine und Angewandte Mathematik}, 134:198,
    1908.

  \bibitem{Weliky1991}
  M.~Weliky, S.~Minsuk, R.~Keller, and G.~Oster.
  \newblock {Notochord morphogenesis in \textsl{Xenopus laevis}: Simulation of
    cell behavior underlying tissue convergence and extension}.
  \newblock {\em Development}, 113:1231, 1991.

  \bibitem{WelikyOster1990}
  M.~Weliky and G.~Oster.
  \newblock {The mechanical basis of cell rearrangement. I.~Epithelial
    morphogenesis during \textsl{Fundulus} epiboly}.
  \newblock {\em Development}, 109:373, 1990.

  \bibitem{WignerSeitz1933}
  E.~Wigner and F.~Seitz.
  \newblock On the constitution of metallic sodium.
  \newblock {\em Physical Review}, 43:804, 1933.

  \bibitem{Young2000}
  B.~Young and J.W.~Heath, editors.
  \newblock {\em Wheater's Functional Histology: A Text and Colour Atlas}.
  \newblock Churchill Livingstone, 2000.

  \bibitem{Zahm1997}
  J.-M.~Zahm, H.~Kaplan, A.-L.~H\'erard, F.~Doriot, D.~Pierrot, 
    P.~Somelette, and E.~Puchelle.
  \newblock Cell migration and proliferation during the in vitro wound repair of
    the respiratory epithelium.
  \newblock {\em Cell Motility and the Cytoskeleton}, 37:33, 1997.
\end{thebibliography}
\end{document}